\documentclass[a4paper,USenglish]{lipics-v2016-modified}

\pdfoutput=1

\usepackage[utf8]{inputenc}
\inputencoding{latin2}

\usepackage{style/allLIPICS}

\title{
    Summaries for Context-Free Games
}

\author[1]{Luk\'a\v s Hol\'ik}
\author[2,3]{Roland Meyer}
\author[2]{Sebastian Muskalla}

\affil[1]{
    Brno University of Technology,
    \texttt{holik@fit.vutbr.cz}
}
\affil[2]
{
    TU Braunschweig,
    \texttt{\{meyer,muskalla\}@cs.uni-kl.de}
}
\affil[3]
{
    Aalto University
}

\titlerunning{Summaries for Context-Free Games}

\authorrunning{L. Hol\'ik, R. Meyer, and S. Muskalla}

\Copyright{Luk\'a\v s Hol\'ik, Roland Meyer, and Sebastian Muskalla}


\newcommand{\gex}{G_{\mathit{ex}}}
\newcommand{\aex}{A_{\mathit{ex}}}

\begin{document}

\maketitle

\vspace{-4mm}
\begin{abstract}
%{\color{red}We should rewrite this to fit the new story}
%We study two-player games played on the infinite graph of sentential forms induced by a context-free grammar (that comes with an ownership partitioning of the non-terminals). 
%The winning condition is inclusion in a regular language for the terminal words derived in maximal plays. 
%Our contribution is an algorithm to compute a finite representation of all plays starting in a non-terminal. 
%The game representation is compositional and, once obtained for the non-terminals, immediate to lift to sentential forms. 
%The result has three consequences. 
%It is decidable whether a position is in the winning region of a player.  
%The winning regions have a finite representation. 
%We can compute a winning strategy. 
%Technically, the algorithm is a fixed-point iteration over a novel domain that generalizes the transition monoid to negation-free Boolean formulas.
%It takes doubly exponential time to check whether a winning strategy exists, which is proven to be the optimal time complexity.
%%
%We show that this domain is not only expressive but also algorithmically appealing. 
%It is compatible with recent antichain- and subsumption-based  optimizations, and admits a lazy evaluation strategy. 
%%
%Our experimental comparison with the algorithm of Cachat \cite{Cachat2002} showed very encouraging results.
%
We study two-player games played on the infinite graph of sentential forms induced by a context-free grammar (that comes with an ownership partitioning of the non-terminals). 
The winning condition is inclusion of the derived terminal word in the language of a finite automaton. 
Our contribution is a new algorithm to decide the winning player and to compute her strategy.
It is based on a novel representation of all plays starting in a non-terminal. 
The representation uses the domain of Boolean formulas over the transition monoid of the target automaton. 
The elements of the monoid are essentially procedure summaries, 
and our approach can be seen as the first summary-based algorithm for the synthesis of recursive programs.
We show that our algorithm has optimal (doubly exponential) time complexity, that it is compatible with recent antichain optimizations, and that it admits a lazy evaluation strategy. 
Our preliminary experiments indeed show encouraging results, indicating a speed up of three orders of magnitude over a competitor. 
\end{abstract}
\vspace{-0.3cm}

% intro
\section{Introduction}
The motivation of our work is to generalize the language-theoretic approach to verification of recursive programs~\cite{HeizmannHoenickePodelski2010,LanguageRefinement} to synthesis.
Central to verification are queries 
$\lang{G}\subseteq \lang{A}$, where $G$ is a context-free grammar representing the control-flow of a recursive program and $A$ is a finite automaton representing the specification.  
When moving to synthesis, 
we replace the inclusion query by a strategy synthesis for an \emph{inclusion game}. 
This means $G$ comes with an ownership partitioning of the non-terminals. 
It induces a game arena defined by the sentential forms and the left-derivation relation (replace the leftmost non-terminal, corresponds to executing the recursive program). 
The winning condition is inclusion in a regular language given by a finite automaton $A$. 
To be precise, player \emph{prover} tries to meet the inclusion by deriving terminal words from the language or enforcing infinite derivations. The goal of \emph{refuter} is to disprove the inclusion by deriving a word outside $\lang{A}$.  

For the verification of recursive programs, the two major paradigms are \emph{summarization} \cite{SharirPnueli1978,RepsHorwitzSagiv1995} and \emph{saturation}~\cite{BouajjaniEM:1997,FinkelWW:1997}.
Procedure summaries compute the effect of a procedure in the form of an input-output relation. 
Saturation techniques compute the pre$^*$-image over the configurations of a pushdown system (including the stack).
%For the verification of recursive programs, there are two competing paradigms. 
%\emph{Procedure summaries} \cite{SharirPnueli1978,RepsHorwitzSagiv1995} summarize the effect of a procedure in the form of an input-output relation. 
%\emph{Saturation techniques} \cite{BouajjaniEM:1997,FinkelWW:1997} compute the $\mathit{pre^*}$-image over the configurations of a pushdown system.
Both were extensively studied, optimized, and implemented~\cite{MOPED,WALi,Beyer:2015,Beyer:2016}.
What speaks for summaries is that they seem to be used more often, as witnessed by the vast majority of verification tools participating in the software verification competition \cite{Beyer:2015,Beyer:2016}.
The reason, besides simpler implementability, may be that the stack maintained by the pre$^*$-construction increases the search space. 
%has a negative impact on the size of the search space.
% without actually influencing the verification task.

Saturation has been lifted to games and synthesis in \cite{Cachat2002,Hague2009}, from which closest to our setting is the work of Cachat \cite{Cachat2002}, 
where the game arena is defined by a pushdown system and
\begin{wraptable}[4]{r}{7.2cm}
\vspace{-0.27cm}
\scalebox{0.92}{
      \begin{tabular} {c|cc}
  
        Problem $\backslash$ Method
        & Saturation
        & Summarization
        \\
        \hline
        Verification
        & \cite{BouajjaniEM:1997,FinkelWW:1997}
        & \cite{SharirPnueli1978,RepsHorwitzSagiv1995}
        \\
        \hline
               Synthesis
        & \cite{Cachat2002,Muscholl2005,Hague2009}  
        & 
        \\
    \end{tabular}}
%    \caption{State-of-the-art.\label{Table:Contribution}}
\vspace{-0.5cm}
    \end{wraptable} 
the winning condition is given by a regular set of goal configurations, and the work of 
Muscholl, Schwentick, and Segoufin \cite{Muscholl2005}, where a problem similar to ours is solved by a reduction to \cite{Cachat2002}.
% "almost identical" is maybe said a little bit too much
%
In this paper, we fill in the empty spot in the picture and propose a solver and synthesis method for context-free inclusion games based on summaries. 

\paragraph*{Overview of Our Method}

Our main contribution is a novel representation of inclusion games that combines well with efficient methods from algorithmic verification (see below). The basic data structure are the elements of the \emph{transition monoid} of the automaton $A$, called boxes.
Boxes are relations over the states of $A$ that capture the state changes on $A$ induced by terminal words~\cite{Buchi1990}. 
As such, they correspond to procedure summaries. 
The set of all plays starting in a non-terminal yields a (typically infinite) tree.  
We show how to represent this tree by a (finite) \emph{negation-free Boolean formula over the transition monoid}, where
 conjunction and disjunction represent the behavior of the players on the inner nodes. 
%

%To compute the representation of the game, 
To compute the representation, 
we employ a fixed-point iteration on a system of equations that reflects closely the rules of the grammar (and hence the shape of the tree). 
Indeed, we simultaneously compute the formulas for all non-terminals. 
In the fixed-point computation, a strategy of prover to enforce an infinite play naturally yields a formula equivalent to $\ifalse$. 
For the domain to be finite, we work modulo logical equivalence.
The order %on formulas 
is implication.
Key to the fixed-point computations is the following compositionality: The formula describing the plays from a sentential form $\alpha \beta$ can be obtained by appropriately composing the formulas for $\alpha$ and $\beta$.  
Indeed, since we consider left-derivations, each play starting in $\alpha \beta$ will have a prefix that coincides with a maximal play starting in $\alpha$, followed by a suffix that essentially is a play from $\beta$. 
%The operation of composition is monotonic wrt. implication.
Composition is monotonic \wrt implication

Having a finite representation for the set of plays starting in each non-terminal has several applications. 
With compositionality, we can construct the formulas for all sentential forms.  
This allows us to decide whether a sentential form is in the winning region of a player: We compute the formula and check whether it is rejecting in the sense that refuter can enforce the derivation of a word rejected by the automaton. 
The latter amounts to evaluating the formula under the assignment that sets to $\itrue$ the rejecting boxes.  
%Moreover, we obtain a finite representation of the winning regions for both players.
%
When a sentential form is found to belong to the winning region of a player, we show how to compute a winning strategy, explained here for refuter. 
We transform the formula to conjunctive normal form (CNF). 
On CNFs, we define so-called \emph{choice functions} that select a box from each clause.  
%We define a strategy such that all conform plays to it end in a terminal word represented by one of the chosen boxes. 
We define a strategy such that all conforming plays end in a terminal word represented by a chosen box. 
Instantiating the strategy for a choice function that only picks rejecting boxes (always possible if the initial formula is rejecting) yields a winning strategy for refuter. 

\paragraph*{Complexity and Efficiency}

%Complexity-wise, our approach is equivalent to that of Cachat \cite{Cachat2002}. 
We show that our algorithm is in \textsf{2EXPTIME}, which is tight by \cite{Muscholl2005}.  
Cachat's algorithm is singly exponential and our input instances can be reduced to his with an exponential blow-up, which together also gives a doubly exponential procedure. 
The complexity of the reduction comes from that it must determinize the automaton $A$ \cite{Muscholl2005}.

Our domain is compatible with algorithmic techniques that have proven efficient in a number of applications (see Section~\ref{sec:related}).   
%
%We show how to adapt two heuristics to our fixed-point computation over (formulas over) boxes,
We show how to adapt two heuristics to our fixed-point computation over formulas over boxes,
namely antichains from~\cite{Fogarty:Efficient,Abdulla:Simulation,Abdulla:Advanced}  
and lazy evaluation inspired by~\cite{fiedor:lazy}. 
We also discuss the compatibility of our technique with recent algorithms for the analysis of well-structured systems.
 %
%These techniques help to handle both levels of the exponential complexity of our algorithm.  
%On the other hand, 
%Heuristics thus only have a chance to alleviate the one remaining level of exponential complexity of Cachat's algorithm itself. 
It is not immediate how to use the same heuristics for Cachat's domain of automata.
Moreover, the determinization within the reduction to Cachat's method does not offer much opportunities for optimization, which means there is one level of exponential complexity that is hardly amenable to heuristics. 
%$hence one is left with one  
%
%To demonstrate efficiency of our domain, 

In preliminary experiments, we have compared an implementation of Cachat's saturation-based algorithm with our new summary-based algorithm. 
The benchmarks were generated according to the Tabakov-Vardi random automata model~\cite{TabakovV:2005} that we adapted to grammars.  
The running times of our algorithm were consistently better by three orders of magnitude (without the aforementioneed optimizations). 
This supports our conjecture that keeping the stack has a negative impact on search procedures, and summaries should be preferable. 

\paragraph*{Acknowledgements}

We thank Olivier Serre, Matthew Hague, Georg Zetzsche, and Emanuele D'Osualdo for helpful discussions.
We thank the reviewers for their feedback.
This work was partially supported by the Czech Science Foundation project 16-24707Y,
the IT4IXS: IT4Innovations Excellence in Science project LQ1602, and the BUT project FIT-S-14-2486.

% basics
\section{Inclusion Games on Context-Free Grammars}

A context-free grammar (CFG) is a tuple $G = (N,T,P)$, where $N$ is a finite set of non-terminals, $T$ is a finite set of terminals with $N \cap T = \emptyset$, and $P \subseteq N \times \vartheta$ is a finite set of production rules. 
Here, $\vartheta = (N \cup T)^*$ denotes the set of sentential forms.
We write $X \to \eta$ if $(X, \eta) \in P$.
%We may assume that every non-terminal is the left-hand side of at least one rule.
We assume that every non-terminal is the left-hand side of some rule.
The \emph{left-derivation relation $\leftderive$} replaces the leftmost non-terminal $X$ in $\alpha$ by the right-hand side of a rule. 
Formally, $\alpha \leftderive \beta$ if $\alpha = w X \gamma$ with $w\in T^*$, $\beta = w \eta \gamma$, and there is a rule $X \to \eta\in P$. 
We use $w$ to refer to terminal words (so that a following non-terminal is understood to be leftmost).
We consider CFGs that come with an \emph{ownership partitioning} $N = N_\roundss \dotcup N_\square$ of the set of non-terminals.
We say that the non-terminals in $N_{\player}$ are owned by player $\player\in \players$. 
The ownership partitioning is lifted to the sentential forms ($\vartheta = \vartheta_{\roundss} \dotcup \vartheta_{\square}$) as follows: $\alpha\in \vartheta_{\square}$ if the leftmost non-terminal in $\alpha$ is owned by $\square$, and $\vartheta_{\roundss}=\vartheta\setminus \vartheta_{\square}$. 
In particular, $\round$ owns all terminal words.
Combined with the left-derivation relation, this
yields a game arena.

\begin{definition}
\label{def:grammar_graph}
    Let $G= (N_\roundss \dotcup N_\square,T,P)$ be a CFG with ownership partitioning. 
    The \emph{arena induced by $G$} is the directed graph $(\vartheta_{\roundss} \dotcup \vartheta_{\square}, \leftderive)$.
\end{definition}
A \emph{play} $p=p_0p_1\ldots$ is a finite or infinite path in the arena.
Being a path means $p_i \leftderive p_{i+1}$ for all positions.  
If it is finite, the path ends in a vertex denoted $p_{\mathit{last}}\in \vartheta$.
A path corresponds to a sequence of left-derivations,   
where for each leftmost non-terminal the owning player selects the rule that should be applied.
A play is \emph{maximal} if it has infinite length or if the last position is a terminal word.

The winning condition of the game is defined by inclusion or non-inclusion in a regular language (depending on who is the player) for the terminal words derived in maximal plays. 
If the maximal play is infinite, it does not derive a terminal word and satisfies inclusion.
The regular language is given by a (non-deterministic) finite automaton \mbox{$A = (T, Q, q_0, Q_F, \to)$}.
Here, $T$ is a finite alphabet, $Q$ is a finite set of states, $q_0 \in Q$ is the initial state, $Q_F \subseteq Q$ is the set of final states, and $\to \ \subseteq Q \times T \times Q$ is the transition relation.
Instead of $(q,a,q') \in\ \to$, we write $q \overset{a}{\to} q'$ 
 and extend the relation to words: $q \overset{w}{\to} q'$ means there is a sequence of states starting in $q$ and ending in $q'$ labeled by $w$. 
The language $\lang{A}$ consists of all words $w\in T^*$ with $q_0 \tow{w} q_f$ for some $q_f\in Q_F$. 
We write $\overline{\lang{A}}=T^*\setminus\lang{A}$ for the complement language.

From now on, we use $A=(T, Q, q_0, Q_F, \to)$ for finite automata and \mbox{$G=(N_\roundss \dotcup N_\square, T, P)$} for grammars with ownership.
Note that both use the terminal symbols $T$.

\begin{definition}
\label{def:inclusion_game}
    The \emph{inclusion game} and the \emph{non-inclusion game} \wrt $A$ on the arena induced by $G$ are defined by the following winning conditions.
    A maximal play $p$ satisfies the \emph{inclusion winning condition} if it is either infinite or we have  $p_{\mathit{last}} \in \lang{A}$.
    A maximal play satisfies the \emph{non-inclusion winning condition} if it is finite and $p_{\mathit{last}} \in \overline{\lang{A}}$.
\end{definition} 
The two games are complementary: For every maximal play, exactly one of the winning conditions is satisfied.
We will fix player $\round$ as the \emph{refuter},
the player wanting plays to satisfy non-inclusion, which is a reachability condition.
The opponent $\square$ is the \emph{prover}, wanting plays to satisfy inclusion, which is a safety condition. 
Since refuter has a single goal to achieve and has to enforce termination, we will always explain our constructions from refuter's point of view.
To win, prover just has to ensure that she stays in her winning region.  
She does not need to care about termination.

A \emph{strategy} for player $\player \in \players$ is a function that takes a non-maximal play $p$ with $p_{\mathit{last}}\in\vartheta_{\player}$ (it is $\player$'s turn) and returns a successor of this last position. 
A play \emph{conforms} to a strategy if whenever it is the turn of $\player$, her next move coincides with the position returned by the strategy. 
A strategy is \emph{winning from a position $p_0$} if every play starting in $p_0$ that is conform to the strategy eventually satisfies the winning condition of the game.
The \emph{winning region} for a player is the set of all positions from which the player has a winning strategy.

\begin{example}
\label{ex:running}
    Consider the grammar
    %\[
    %    \gex=(\set{X,Y}, \set{a,b}, \set{X \to aY, X \to \varepsilon, Y \to bX})
    %    \ .
    %\]
    $
        \gex=(\set{X,Y}, \set{a,b}, \set{X \to aY, X \to \varepsilon, Y \to bX})
        \ .
    $
    The automaton $\aex$ is given in Figure~\ref{fig:automaton_boxes} and accepts $(ab)^*$.
    If refuter owns $X$ and prover owns $Y$, then prover has a winning strategy for the inclusion game from position $X$.
    Indeed, finite plays only derive words in $(ab)^*$.
    Moreover, if refuter enforces an infinite derivation, prover wins inclusion as no terminal word is being derived.
    Refuter can win non-inclusion starting from $Y$.
    After prover has chosen $Y \to bX$, refuter selects $X \to \varepsilon$ to derive $b \not\in (ab)^*$. 
\qed
\end{example}
Our contribution is an algorithm to compute (a representation of) both, the winning region of the non-inclusion game for $\round$ and the winning region of the inclusion game for $\square$.

% section 3 - from inclusion games to fixed points
\section{From Inclusion Games to Fixed Points}
\label{sec:domain}

We give a summary-based representation of the set of all plays from each non-terminal and a fixed-point analysis to compute it.
We lift the information to the sentential forms. 

%\subsection{Domain}

\subsection{Domain}

The idea of the analysis domain is to use Boolean formulas over words. 
To obtain a finite set of propositions, we consider words equivalent that induce the same state changes on $A$, denoted by $\stateequiv{A}$. 
The winning condition is insensitive to the choice of $\stateequiv{A}$-equivalent words. 
This means it is sufficient to take formulas over
$\stateequiv{A}$-equivalence classes.

To finitely represent the $\stateequiv{A}$-equivalence classes, we rely on the \emph{transition monoid of~$A$}, defined as $\Boxes_A = (\pwrset{Q \times Q},\ ;\ ,\ \id)$.
We refer to the elements $\tboxr, \tboxt \in \Boxes_A$ as \emph{boxes}.
Since boxes are relations over the states of $A$, their relational composition is defined as usual, 
$
%\[
    \tboxr;\tboxt
    =
    \Set{ (q,q'') }
    {
        \exists q' \in Q:
        (q, q') \in \tboxr\text{ and } (q', q'') \in \tboxt
    }
    .
%\]
$  
Relational composition is associative. 
The identity box $\id = \Set{ (q,q) }{ q \in Q}$ is the neutral element \wrt relational composition.  

A box $\tboxr$ represents the language
$\lang{\tboxr} = \{ w\in\lang{\tboxr} \ \mid \ \forall q,q' \in Q \colon q\xrightarrow w q' \text{ iff } (q, q')\in\tboxr \}$.
That is, the words induce
%precisely % I would like to write this, but it breaks the formatting
the state changes specified by the box. 
Hence, $\lang{\tboxr}$ is an equivalence class of $\stateequiv{A}$, finitely represented by $\tboxr$.
The function  $\tbox{-} : T^* \to \Boxes_A$ maps $w$ to the unique box $\tbox{w}$ representing the word, $w \in \lang{\tbox{w}}$\,.
More explicitly, $\tbox{\varepsilon} = \id$, 
$\tbox{a} = \smallset{ (q,q') }{ q \overset{a}{\to} q' }$ for all $a\in T$, and $\tbox{uv}=\tbox{u};\tbox{v}$. 
The image $\tbox{T^*}$ contains exactly the boxes $\tboxr$ with $\lang{\tboxr}\neq \emptyset$. 
Figure~\ref{fig:automaton_boxes} illustrates the representation of words as boxes.
\begin{figure}[t]
\vspace{-2mm}
%\begin{figure}[ht]
    {
        \begin{minipage}{.3\textwidth}
            \hspace{3mm}
            \scalebox{0.9}{\begin{tikzpicture}[scale=0.8,->,>=stealth',shorten >=1pt,auto,node distance=2.8cm,
semithick]
%\tikzstyle{every state}=[fill=red,draw=none,text=white]

\node[initial,state,accepting,transform shape, initial text={}] (A)                    {$q_0$};
\node[state, transform shape] (B) [right of=A] {$q_1$};

\path
(A) edge [bend left = 15] node [above]  {a} (B)
(B) edge [bend left = 15] node [below] {b} (A);
\end{tikzpicture}}
        \end{minipage}
        \hspace*{0.6cm}
        \begin{minipage}{.6\textwidth}
            \scalebox{0.8}{\begin{tikzpicture}[scale=0.5]
    \pic  (fid) [transform shape] {tbox={2}{1/1/,2/2/}};
    \node [below = 0.6cm] {$\fid = \tbox{\varepsilon}$};
    
    \pic (a) at (3.5,0) [transform shape] {tbox={2}{1/2/}};
    \node [below = 0.6cm] at (3.5,0) {$\tbox{a}$};
    
    \pic (b) at (7,0) [transform shape] {tbox={2}{2/1/}};
    \node [below = 0.6cm] at (7,0) {$\tbox{b}$};
    
    \pic (ab) at (10.5,0) [transform shape] {tbox={2}{1/1/}};
    \node [below = 0.6cm] at (10.5,0) {$\tbox{ab}$};
    
    \pic (ba) at (14, 0) [transform shape] {tbox={2}{2/2/}};
    \node [below = 0.6cm] at (14, 0) {$\tbox{ba}$};
    
    \pic (aa) at (17.5, 0) [transform shape] {tbox={2}{}};
    \node [below = 0.6cm] at (17.5,0) {$\tbox{aa} = \tbox{bb}$};
\end{tikzpicture}}
        \end{minipage}
    }
    \vspace*{-0.3cm}
    \caption{The automaton $\aex$ accepting $(ab)^*$ and all its boxes with non-empty language. The first dash on each side of a box represents state $q_0$, the second dash represents $q_1$.}
    \label{fig:automaton_boxes}
\vspace{-3mm}
\end{figure}
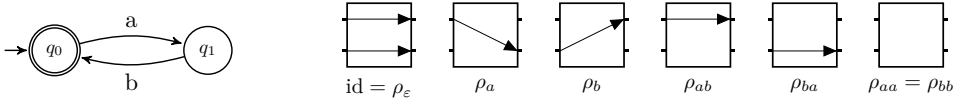

The terminal words generated by maximal plays are represented by boxes, disjunction gives the alternatives of refuter, and conjunction expresses the options for prover. 
The set of plays from a given position is thus represented by a formula $F$ from the set $\BF_A$ of \emph{negation-free Boolean formulas over the transition monoid} (propositions are boxes).
This set includes the unsatisfiable formula $\ifalse$. We use the rules $\ifalse \wedge F = F \wedge \ifalse = \ifalse$ and $\ifalse \vee F = F \vee \ifalse = F$ to evaluate conjunctions and disjunctions involving $\ifalse$ on the syntactical level. As a consequence, $\ifalse$ is the only syntactic representation of the unsatisfiable formula. 
This will simplify the definition of relational composition.
From now on and without further mentioning, $F$ and $G$ will refer to formulas from $\BF_A$.

Our goal is to decide whether refuter can force the plays from an initial position to end in a terminal word rejected by $A$. 
To mimic this, we define a formula to be \emph{rejecting} if it is satisfied under the assignment $\nu:\Boxes_A\rightarrow \{\itrue,\ifalse\}$ such that  
$\nu(\rho) = \itrue$ if and only if $\rho$ does not contain a pair $(q_0, q_f)$ with $q_f\in Q_F$.

To use formulas in a Kleene iteration, we have to define a partial ordering on them. 
Intuitively, $F$ should be smaller than $G$ if $G$ makes it easier for refuter to win. 
Taking the logical perspective, %where we evaluate the rejecting assignment bottom-up, 
it is easier for refuter to win if $F$ implies $G$. 
Implication on $\BF_A$ is not antisymmetric.
To factor out the symmetries, we reason modulo logical equivalence, $\BFE$. 
Every formula is understood as a representative of the class of logically equivalent formulas of $\BF_A$. 
%Let $\BFE$ be the set of equivalence classes of $\BF_A$.
Extending $\lleq$ to $\BFE$ by comparing representatives then 
yields a partial order.
The least element of the partial order is the equivalence class of $\ifalse$.

%\subsection{Operations}

\subsection{Operations}
\label{sec:operations}

We combine formulas by conjunction, disjunction, and by an operation of relational composition that lifts $;$ from the transition monoid to formulas over boxes.
To explain the definition of relational composition, note that every finite maximal play from $\alpha\beta$ proceeds in two phases. 
It starts with a maximal play turning $\alpha$ into a terminal word, say $w$, followed by a play from $w\beta$. 
Since there are no more derivations for $w$, the play from $w\beta$ coincides with a play from~$\beta$, except that all sentential forms have a prefix $w$.

Let $F$ and $G$ represent all plays starting in $\alpha$ and $\beta$, respectively. 
In $F$, terminal words like $w$ are represented by boxes $\rho$. 
We append the plays from $\beta$ by replacing $\rho$ with $\rho;G$.
To take into account all plays from $\alpha$, we do this replacement for all boxes in $F$. 
It remains to add the prefix $w$ to the sentential forms
in the plays from $\beta$. 
In $\rho;G$, every box $\tboxt$ in $G$ is replaced by $\tboxr;\tboxt$. 
The so-defined formula $F;G$ will represent all plays from $\alpha\beta$. 

\raggedbottom
\newpage

\begin{example}
    \newcommand {\ta}
    { \tboxr_a }
    \newcommand {\tb}
    { \tboxr_b }
    \newcommand {\tc}
    { \tboxr_c }
    \newcommand {\td}
    { \tboxr_d }
Let $F =\ta \vee \tb$ and $G =  \tc \wedge \td$.
    We have
$
    (\ta \vee \tb);(\tc \wedge \td)=  \ta;(\tc \wedge \td)  \vee \tb;(\tc \wedge \td) =  (\ta;\tc \wedge \ta;\td)  \vee (\tb;\tc \wedge \tb;\td)\ .
 $
%%   \begin{align*}
%%     (\ta \vee \tb);(\tc \wedge \td)&=  \ta;(\tc \wedge \td)  \vee \tb;(\tc \wedge \td) =  (\ta;\tc \wedge \ta;\td)  \vee (\tb;\tc \wedge \tb;\td)\ .
%%         \end{align*}
The first equality replaces $\ta$ and $\tb$ by $\ta;G$ and $\tb;G$, respectively. 
The second equality prefixes $\tc$ and $\td$ in $G$ by the corresponding box $\ta$ or $\tb$.\tqed
\end{example}

\begin{definition}
\label{def:cnf_seqcomp}
    \emph{Relational composition} over $\BF_A$ is defined by $F;\ifalse = \ifalse;G = \ifalse$ and for composite formulas ($\star\in\{\land,\lor\}$, $\rho\in\Boxes_A$) by
%%     \begin{align*}        
%%         (F_1\star F_2) ; G =  F_1 ; G \star F_2 ; G\hspace{2cm} \rho ; (G_1 \star G_2) =\ 
%%         & \rho ; G_1 \star \rho ; G_2\ .
%%     \end{align*}
\begin{align*}
        (F_1\star F_2) ; G =  F_1 ; G \star F_2 ; G\qquad\text{and}\qquad  \rho ; (G_1 \star G_2) =\ 
         \rho ; G_1 \star \rho ; G_2\ .
\end{align*}
\end{definition}
Note that the composition of two non-$\ifalse$ formulas is not $\ifalse$.
Therefore, the result of a relational composition is $\ifalse$ if and only if at least one of the arguments was $\ifalse$.

Relational composition equips the set of formulas with a monoid structure. In particular, relational composition is associative.
For a fixed-point iteration, the operations also have to be monotonic \wrt $\lleq$. 
For conjunction and disjunction, monotonicity obviously holds.

\begin{lemma}
\label{lemma:leq_monotonic}
    If $F \lleq F'$ and $G \lleq G'$, then
    $F;G \lleq F';G'$.
\end{lemma}

\begin{proof}
    The proof proceeds in phases (1) to (4) so that the claim in each phase is proven under the assumption of the claim proven in the previous phase. Let $\set{ \star, \bar{\star} } = \{\land,\lor \}$.
    In the following, we will use $\star$ and $\bar\star$ as syntactic parts of formulas as well as to connect statements in the proof.
    \begin{enumerate}[(1)]
        \item
        First, we prove the lemma for the case when $F,F', G'\in \Boxes_A$
        by induction on the structure of $G$. 
        In the base case, all formulas are boxes, hence $F = F'$ and $G = G'$, and the lemma holds trivially.
        For the induction step, let $G = G_1\star G_2$.
        Note that the Boolean formulas  $(a \star b) \lleq c$ and $(a\lleq c) \mathrel{\bar{\star}} (b \lleq c)$ are equivalent, called Equivalence (i) in the following.
        By (i), we get $(G_1\lleq G')\mathrel{\bar\star}(G_2\lleq G')$.
        Hence, by the induction hypothesis applied twice and by the monotonicity of $\bar{\star}$,
        \mbox{$(F;G_1 \lleq F';G')\mathrel{\bar\star}( F;G_2\lleq F';G')$}.
        Again by Equivalence~(i), we get
        \mbox{$(F;G_1 \star F;G_2) \lleq F';G'$}. 
        This is $F;G \lleq F';G'$ by the definition of relational composition since $F$ is a box.
        \item
        Next, we assume that $F',G'\in \Boxes_A$ and $F$ and $G$ are arbitrary formulas.
        We prove the statement by induction on $F$.
        In the base case, all formulas except $G$ are boxes, hence (1) proves the statement.
        For the induction step, let $F = F_1\star F_2$. 
        By Equivalence (i), 
        we get $(F_1\lleq F')\mathrel{\bar\star}(F_2\lleq F')$. 
        Therefore, by the induction hypothesis and the monotonicity of $\bar{\star}$,
        \mbox{$(F_1;G \lleq F';G')\mathrel{\bar\star}( 
            G_2;M\lleq G';M')$}.   
        This is by (i) equivalent to 
        \mbox{$\left( F_1;G \star F_2;G \right) \lleq F';G'$}. 
        This shows \mbox{$F;M \lleq F';G'$} by the definition of relational composition.
        \item
        Next, we assume only that $F'\in \Boxes_A$ and prove the statement using induction on the structure of $G'$. 
        In the base case, all formulas except $F$ and $G$ are boxes, hence the statement is proven by (2).
        Let $G' = G_1'\star G_2'$. 
        By the general equivalence of the Boolean formulas 
        $a \lleq (b\star c)$ and  $(a \lleq b) \star (a\lleq c)$, called Equivalence (ii) in the following,
        we get $(G \lleq G_1')\mathrel{\star}(G \lleq G_2')$. 
        Therefore, by the induction hypothesis and the monotonicity of $\star$, $(F;G \lleq F';G_1')\mathrel{\star}( F;G\lleq F';G_2')$ holds.
        Again by (ii), we get $F;G \lleq (F';G_1' \star F';G_2')$.  
        This is $F;G \lleq F';G'$ by the definition of relational composition since $F'$ is a box.
        \item
        Finally, we show the general case by induction to the structure of $F'$. 
        In the base case, $F'$ is a box, hence (3) proves the statement.
        Let $F' = F_1'\star F_2'$. 
        By Equivalence~(ii), we get $(F\lleq F_1')\mathrel{\star}(F\lleq F_2')$. 
        Therefore, by the induction hypothesis, and the monotonicity of $\star$, $(F;G \lleq F_1';G')\mathrel{\star}( F;G\lleq F_2';G')$ holds.   
        Again by Equivalence~(ii), we get \mbox{$F;G \lleq (F_1';G' \star F_2';G')$}, which is $F;G \lleq F';G'$ by the definition of relational composition.    
    \end{enumerate}
    \vspace*{-0.5cm}
\end{proof}
We lift the three operations to $\lsim$-equivalence classes by applying them to arbitrary representatives. 
Since implication is transitive, monotonicity of the operations ensures well-definedness. 
Moreover, the operations still behave monotonically on $\BFE$.
From now on, we can thus identify formulas with the classes they represent.

\subsection{System of Equations}

We introduce one variable $\var{X}$ for each non-terminal $X \in N$. 
Terminals $a\in T$ yield boxes, and we write $\var{a}$ for $\tbox{a}$. 
We lift the notation $\var{-}$ to sentential forms: $\var{\varepsilon} = \id$ and $\var{\alpha\beta} = \var{\alpha};\var{\beta}$. 
This means concatenation in rules is replaced by relational composition.
All rules for the same non-terminal are combined into one equation using disjunction or conjunction, depending on who is the owner of the non-terminal.

\begin{definition}
\label{def:game_dfa}
    The \emph{system of equations (over $\BFE$) induced by $G$ and $A$} has one equation for each non-terminal $X \in N_{\player}$ with $\player \in \players$.
    If $X \to \eta_1$, \ldots, $X \to \eta_k$ are all rules with $X$ as their left-hand side, the equation is
    \[
        \var{X}
        =
        \var{\eta_1} \wedge \cdots \wedge \var{\eta_k}\ ,
        \qquad \text{if } X \in N_\square \ ,
    \]
    \[
        \var{X}
        =
        \var{\eta_1} \vee \cdots \vee \var{\eta_k}\ ,
         \qquad \text{if } X \in N_\roundss\ .
    \]
\end{definition}
With Lemma~\ref{lemma:leq_monotonic}, for each non-terminal $X$ we can understand the right-hand side of the associated equation as a monotonic function $f_X : (\BFE)^N \to \BFE$.
It takes as input a vector of formulas (one for each non-terminal)  and computes a new formula for $\var{X}$. 
We combine the functions for each non-terminal to a single function $f : (\BFE)^N \to (\BFE)^N$. 
It is monotonic on the product domain \wrt the product order $\Rightarrow^N$.

Since $\BFE$ with $\Rightarrow$ is a finite bottomed partial order,
there is a unique least solution $\sol{}$ for the equation $\var{} = f ( \var{})$, namely $\sol{} = \bigsqcup_{i \in \N} f^i (\bot)$ \cite{LatticesAndOrder}.
The least element of the product domain is the vector with the $\lsim$-equivalence class of $\ifalse$ in every component. 
Note that the solution is computed by iteratively applying $f$   until a fixed point is reached. 
This procedure terminates since the chain
 \[
     \bot \Rightarrow^N f(\bot) \Rightarrow^N f(f(\bot)) \Rightarrow^N \ldots
 \]
%$
%    \bot \Rightarrow^N f(\bot) \Rightarrow^N f(f(\bot)) \Rightarrow^N \ldots
%$
stabilizes on a finite domain.

The solution $\sol{} : N \to \BFE$  yields a value $\sol{X}$ for each non-terminal $X \in N$.
We lift the notation to sentential forms by $\sol{\varepsilon} = \id$, $\sol{a} = \tbox{a}$ for all $a\in T$, 
and $\sol{ \alpha\beta} = \sol{\alpha};\sol{\beta}$.
From now on, $\sol{}$ will always be the least solution to a system of equations.
The system will be clear from the context (either $G$, $A$ from the development or $\gex$, $\aex$ from the running example).

\begin{example}
\label{ex:sol_game_dfa_ab}
    For $\gex$ and $\aex$ from Example~\ref{ex:running}, 
%%   the system of equations is
%%    \[
%%        \var{X}
%%        = \var{a};\var{Y} \vee \var{\varepsilon}
%%        =  \tbox{a};\var{Y} \vee  \id
%%        \hspace{2cm}
%%        \var{Y}
%%        = \var{b};\var{X}
%%        = \tbox{b};\var{X}\ .
%%   \]
   the system of equations consists of
\mbox{$
        \var{X}
        = \var{a};\var{Y} \vee \var{\varepsilon}
        =  \tbox{a};\var{Y} \vee  \id
$} and $
        \var{Y}
        = \var{b};\var{X}
        = \tbox{b};\var{X}\ .
$
    Its least solution is \mbox{$\sol{X}= \id \lor \tbox{ab}$} and $\sol{Y}=  \tbox{b}$. 
    To terminate the iteration, use $\tbox{bab} = \tbox{b}$.
    \tqed
\end{example}
%
%The main result in this section states that the fixed point $\sol{\alpha}$ is equivalent to the tree $\FT{\alpha}$ of all plays starting in $\alpha$. 
%To be precise, we understand the tree as a (typically infinite)  formula where inner nodes owned by refuter yield disjunctions, prover's nodes are conjunctions, and terminal words are boxes.
%For the proof, we develop machinery for infinite formulas.
%
%\begin{theorem}
%\label{thm:correspondence}
%    $\FT{\alpha} \lsim \sol{\alpha}$.
%\end{theorem}

% semantics
\section{Semantics}
\label{sec:semantics}

Our goal is to determine whether refuter has a winning strategy for the non-inclusion game played from a given sentential form $\alpha$.  
The plays starting in $\alpha$ form a (typically infinite) tree.
The given sentential form is the root and each node has one successor for each sentential form that can be obtained by a left-derivation step. 
This means the leaves are precisely the terminal words derivable in plays from $\alpha$. 
Recall that it is refuters goal to disprove the inclusion in $\lang{A}$. 
We call a leaf \emph{rejecting} if it corresponds to a word outside this language. 
An inner node is rejecting if it is either owned by refuter and it has a rejecting successor, or it is owned by prover and all  successors are rejecting. 
Refuter then has a winning strategy for non-inclusion when playing from $\alpha$ if and only if the root is rejecting.

We understand the tree as an infinite negation-free Boolean formula.  
The terminal words at the leaves are the atomic propositions. 
Each inner node corresponds to an operation of conjunction or disjunction, depending on who is the owner of the sentential form.  
Deciding whether the root is rejecting then amounts to computing the truth value of this infinite formula under the assignment that sets to true precisely the words outside $\lang{A}$.

Our goal is to compare this infinite formula with the finite formula obtained as the least solution of the fixed-point iteration presented in Section \ref{sec:domain}.

\subsection{Emptiness Games}

As a preparation, we note that the fixed-point solution to the inclusion game in particular solves the so-called \emph{emptiness game}, where a maximal play is winning if it is infinite. 
The emptiness game can be understood as the inclusion game \wrt $\lang{A} = \emptyset$. 
A winning strategy of prover for the emptiness game is thus a strategy that only generates infinite plays.   
We show that such a strategy exists when playing from $\alpha$ if and only if $\sol{\alpha} = \ifalse$. 

To prove the equivalence, we will also show that if $\sol{\alpha} \neq \ifalse$ refuter has a strategy to enforce finite plays. 
To define this strategy, we use the following notation.
Let $\sol{}^i$ be the $i^{\text{th}}$ Kleene approximant of the least solution to the system of equations, so $\sol{X}^i$ is the $i^{\text{th}}$ approximation to the value of the non-terminal $X$. 
We define $\sol{a}^i = \tbox{a}$ for $a \in T \cup \set{ \varepsilon }$ and all $i \in \N$. 
Just as we did for the fixed-point solution, we inductively define $\sol{\alpha.\beta}^i = \sol{\alpha}^i ; \sol{\beta}^i$.

Note that we deal with negation free-formulas and evaluate conjunctions and disjunctions involving $\ifalse$ on the syntactic level.
As a consequence, the result of a conjunction is $\ifalse$ if and only if at least one of the conjuncts was $\ifalse$ and the result of a disjunction is $\ifalse$ if and only if both disjuncts were $\ifalse$.

\begin{theorem}
    Prover has a winning strategy for the emptiness game iff \mbox{$\sol{\alpha} = \ifalse$}.
\end{theorem}

\begin{proof}
    We have $\sol{X}^0 = \ifalse$ for every non-terminal $X$.
    If $i_0$ is a number such that $\sol{X}^{i_0} \neq \ifalse$, then $\sol{X}^{i_0+k} \neq \ifalse$ for all $k \in \N$.
    This follows from monotonicity of the solution to the system of equations, $\sol{X}^{i} \lleq \sol{X}^{i+k}$ for all $i, k \in \N$, and from the fact that $F \lleq \ifalse$ implies $F \lsim \ifalse$. 
    Note that for a terminal word, we have $\sol{w}^i \neq \ifalse$ for all $i \in \N$.
    Therefore, $\sol{\beta} \neq \ifalse$ holds if and only if $\sol{X} \neq \ifalse$ for all non-terminals $X$ occurring in $\beta$ by the definition of relational composition.
    
    Assume $\sol{\alpha} = \ifalse$.
    We define a strategy $s_{\infty}$ for prover such that all plays conform to it satisfy $\sol{\beta} = \ifalse$ for all positions $\beta$ occurring in the play. 
    Since we have $\sol{w} \neq \ifalse$ for terminal words, this means the plays have to be infinite. 
        
    Assume $wX\beta$ is given and it is prover's turn. 
    If there is $Y$ in $\beta$ so that $\sol{Y} = \ifalse$, we can pick any rule $X \to \eta$ and the result will still satisfy $\sol{w \eta \beta} = \ifalse$. 
    If there is no such non-terminal, we know that $\sol{X} = \ifalse$ has to hold. 
    The fixed-point solution to the system of equations satisfies 
    $\sol{X} = \bigwedge_{X \to \eta} \sol{\eta}$.  
    Since the conjunction is $\ifalse$, there is at least one rule $X \to \eta$ with $\sol{\eta} = \ifalse$.  
    If we pick this rule, the resulting position will satisfy $\sol{w \eta \beta} = \ifalse$. 
        
    Assume $wX\beta$ is given and it is refuters's turn. 
    If there is a non-terminal $Y$ in $\beta$ with $\sol{Y} = \ifalse$, the position resulting from any rule will still satisfy $\sol{w \eta \beta} = \ifalse$. 
    If there is no such position, we know that $\sol{X} = \ifalse$. 
    The fixed-point solution to the system of equations satisfies
        $\sol{X} = \bigvee_{X \to \eta} \sol{\eta}$. 
    Since the disjunction is $\ifalse$, for any rule $X \to \eta$ we get $\sol{\eta} = \ifalse$. 
    Therefore, $\sol{w \eta \beta} = \ifalse$ has to hold, no matter which rule $X \to \eta$ refuter chooses.
        
    For the other direction, we show that whenever $\sol{\alpha} \neq \ifalse$ refuter has a strategy such that all plays conform to it are finite (a winning strategy for the non-emptiness game). 
    In this case, prover cannot have a strategy that only generates infinite plays.
        
    We define the \emph{$i$-step attractor} $\Attr_i$ to be the set of non-terminals $X$ such that $\sol{X}^i \neq \ifalse$.
    Note that the $i$-step attractors form a chain $\Attr_0 \subseteq \Attr_1 \subseteq \ldots$\ .
    The chain has to stabilize because the set of non-terminals is finite.
    Given a non-terminal $X$, we define its \emph{attractor level} to be the lowest index $i$ such that $X \in \Attr_i$, and $\infty$ if no such index exists.
    We define the attractor level to be $0$ for terminal symbols $a \in T$.
\newcommand{\absG}{|G|}
    The \emph{attractor level} of a sentential form $\beta$ is defined to be
    \[
        \level(\beta) = \sum_{j = 1, \ldots, |\beta|} {\absG}^{\textit{attractor level of } \beta_j}
        \ .
    \]
    We give a strategy for refuter such that for each play conform to it, the levels of the occurring positions form a strictly decreasing chain. 
    Since such a chain has to be finite, this proves the claim.
        
    Consider $wX\beta$ with $\sol{wX\beta} \neq \ifalse$.
    In particular, $\sol{X} \neq \ifalse$ and $\sol{Y} \neq \ifalse$ for all non-terminals $Y$ in $\beta$. 
    This means $\level{wX\beta} \neq \infty$.
        Let $i$ be the attractor level of $X$.
        
    Assume it is prover's turn.
    We have $\sol{X}^i = \bigwedge_{X \to \eta} \sol{\eta}^{i-1}$. 
    Since $X$ is in the $i$-step attractor, we get $\sol{X}^i \neq \ifalse$ and thus $\sol{\eta}^{i-1} \neq \ifalse$ for all rules $X \to \eta$.
    Hence, any symbol occurring in the right-hand side $\eta$ of a rule for $X$ has attractor level at most $i-1$. 
    Applying such a rule replaces $X$, which contributes ${\absG}^i$ to the level of $w X \beta$, by $\eta$, a sequence with
    \[
            \level{\eta}
            \leq |\eta| {\absG}^{i-1}
            < |G| {\absG}^{i-1}
            = {\absG}^{i}\ .
    \]
    We conclude $\level{w X \beta} > \level{ w \eta \beta}$ for any rule $X \to \eta$.
        
    Assume it is refuter's turn. 
    We have $\sol{X}^i = \bigvee_{X \to \eta} \sol{\eta}^{i-1}$.
    Since $\sol{X}^i \neq \ifalse$, there is a rule $X \to \eta$ with $\sol{\eta}^{i-1} \neq \ifalse$. 
    In particular, all symbols in $\eta$ have attractor level at most $i-1$. 
    If we pick the rule $X \to \eta$, we get $\level{w X \beta} > \level{w \eta \beta}$.
\end{proof}

\subsection{The Infinite Tree of Plays vs. the Fixed-Point Solution}
\label{subsec:semantics}

We discuss how the fixed-point solution to the system of equations relates to the (typically infinite) formula representing the tree of all plays from a given position.
Let $\FT{\alpha}$ be the tree of plays from $\alpha$. 
We argued that we can understand $\FT{\alpha}$ as an infinite negation-free Boolean formula, where inner nodes are disjunctions or conjunctions (logical connectives) and leaves are boxes (atomic propositions).   
We call $\FT{\alpha}$ a \emph{formula tree} to emphasize the fact that we can see it as both, a formula and a tree. 
Note that the outdegree of inner nodes in $\FT{\alpha}$ is bounded by the maximal number of rules for each non-terminal.  
We identify the empty formula tree with $\ifalse$.
To generalize the notions of assignment and value from finite formulas to infinite formula trees, we need a least fixed point that propagates the values from the leaves up the tree.
The following constructions are conservative extensions of the finite case and behave as expected when applied to finite formulas.

Given a formula tree $\FT$, an \emph{evaluation} $e$ is a map from the nodes to $\set{\itrue,\ifalse}$.  
We understand an evaluation as the set of nodes with value $\itrue$. 
This helps us see that the set of evaluations on a fixed tree $\FT$ ordered by inclusion forms a complete lattice: The least element is the empty set, the join is the union, and the meet is the intersection.

Given $e$, we define the \emph{1-step propagation} $p(e)$ to be the evaluation obtained from $e$ by adding 
(1) all disjunctions $n$ so that some successor of $n$ is in $e$ and
(2) all conjunctions $n$ so that all successors of $n$ are in $e$.  
Note that $e \subseteq p(e)$, so $p$ is monotone over tree evaluations.

An assignment of boxes $\nu : \Boxes_A \to \set{\itrue, \ifalse}$ induces an evaluation $e_\nu$
on a formula tree over $\Boxes_A$.  
The leaves $\tboxr$ evaluate to $\nu(\tboxr)$, all other nodes evaluate to $\ifalse$. 
We define the \emph{propagation} of $\nu$ on $\FT$ to be 
the join $\bigsqcup_{i \in \N} p^i (e_\nu)$.
The \emph{value} of $\FT$ under $\nu$ is the value of the root node in the propagation of $\nu$ on $\FT$.
We define implication as usual: $\FT \lleq \FT'$ if under all $\nu$ the value of $\FT'$ is at least the value of $\FT$.

\begin{lemma}
\label{lemma:index_true}
    The value of $\FT$ under $\nu$ is $\itrue$ iff the root of $\FT$ is in $p^{i_0} (e_\nu)$, for some $i_0\in \N$.
\end{lemma}

\begin{proof}
    If $i_0$ exists, the claim follows from 
    $p^{i_0} (e_\nu) \subseteq \bigsqcup_{i \in \N} p^i (e_\nu)$. 
    Assume for all $i$, the root is not in $p^i (e_\nu)$. Since the join is the union, the root node will not be in the propagation.
\end{proof}
Given a formula tree $\FT$, a subset of its nodes $S$ induces the \emph{prefix} $\FT^{S}$ defined by removing from $\FT$ the subtrees with root in $S$. 
We think of a removed subtree as being evaluated to $\ifalse$. 
We simplify the formula by propagating $\ifalse$ upwards from the removed subtrees. 
Formally, $\FT{}{S}$ is created by
\mbox{(1) marking} all subtrees with root in $S$, \
(2) repeating until fixed point:
marking all disjunctions that have all successors marked 
and 
marking all conjunctions that have some successor marked, including their subtrees,
(3) deleting all marked nodes. 
$\FT^{S}$ is again a formula tree, each disjunction and conjunction has at least one successor. 

Since we deal with negation-free formulas, the removed subtrees  will never lead to the value of a tree being $\itrue$ that would have been $\ifalse$ prior to the removal. 
Hence, a larger prefix will be easier to satisfy than a smaller one (created by removing more nodes). 

\begin{lemma}
\label{lemma:remove_superset_implies}
    Let $\FT$ be a formula tree and let $S \subseteq R$ be subsets of its nodes.
    Then $\FT{}{R} \lleq \FT{}{S}$.
\end{lemma}

\begin{proof}
The nodes of $\FT{}{R}$ form a subset of the nodes of $\FT{}{S}$.
Hence, for all assignments $\nu$ we have 
    $e_\nu (\FT{}{R}) \subseteq e_\nu (\FT{}{S})$ and
    $p^i\left(e_\nu (\FT{}{R})\right) \subseteq p^i\left(e_\nu (\FT{}{S})\right)$
    for all $i$.
    In the limit, this yields 
    $
        \bigsqcup_{i \in \N} p^i\left(e_\varphi (\FT{}{R})\right)
        \subseteq
        \bigsqcup_{i \in \N} p^i\left(e_\varphi (\FT{}{S})\right)
    $.
    In particular, the value of the root node of $\FT{}{S}$ is at least the value of the root node of $\FT{}{R}$.
\end{proof}
Let $S_i$ be the set of nodes in $\FT$ of depth strictly greater than $i$. 
We write $\FT{}{i}$ for $\FT{}{S_i}$ and call this tree the \emph{cut at level $i$}. 
The previous lemma implies $\FT{\alpha}{i} \lleq \FT{\alpha}{i+k}$ for all $i, k \in \N$. 
The value of the whole tree $\FT$ is $\itrue$ if and only if there is a finite cut such that the value of the prefix is already $\itrue$.

\begin{lemma}
    \label{lemma:index_true_cut}
    The value of $\FT$ under $\nu$ is $\itrue$ iff the value of $\FT{}{i_0}$ is $\itrue$, for some $i_0\in \N$.
\end{lemma}

\begin{proof}
    By Lemma \ref{lemma:index_true}, it is sufficient to show that for all $i_0 \in \N$, the root of $\FT$ is in $p^{i_0} (e_\nu)$ if and only if the value of $\FT{}{i_0}$ is $\itrue$.
        
    Assume the root of $\FT$ is in $p^{i_0} (e_\nu)$. 
    Since each step of propagation can only propagate the value of a node to its immediate predecessors, there is a set of leaves of height at most $i_0$ assigned to $\itrue$ by $\nu$ that are sufficient to cause the root to be in $p^{i_0} (e_\nu)$.
    Consider the evaluation $e'$ that coincides with $e_\nu$ on nodes up to height at most $i_0$ and evaluates all other nodes to $\ifalse$. 
    We obtain $e' \subseteq e_\nu$, and by the monotonicity of propagation $p^{i_0} (e') \subseteq p^{i_0} (e_\nu)$. 
    Since $\itrue$ got propagated to the root by leaves of depth $i_0$ and those leaves are still evaluated to $\itrue$ by $e'$, the root is contained in $p^{i_0} (e')$. 
    To finish the proof, note that the semantics of removing nodes of depth greater than $i_0$ is defined by propagating $\ifalse$ from nodes of depth greater than $i_0$. This coincides with the behavior of $p^{i_0} (e')$.
    
    Assume the value of $\FT{}{i_0}$ is $\itrue$, which means the root of $\FT{}{i_0}$ is in
    $\bigsqcup_{i \in \N} p^{i}( e_\nu ( \FT{}{i_0} ))$. Since $\FT{}{i_0}$ has height at most $i_0$, we know
    $\bigsqcup_{i \in \N} p^{i}( e_\nu ( \FT{}{i_0} )) = p^{i_0}( e_\nu ( \FT{}{i_0} ))$.
    By $e_\nu (\FT{}{i_0}) \subseteq e_\nu (\FT{}{})$ we get 
    $p^{i_0}\left(e_\nu (\FT{}{i_0})\right) \subseteq p^{i_0}\left(e_\nu (\FT{}{})\right)$, as desired.
\end{proof}
The next lemma shows that the composition of two trees for sentential forms, each cut at some level, is a prefix of the tree for the concatenation of the sentential forms, cut at the sum of levels.
The composition appends the second tree to every leaf of the first, as before.

\begin{lemma}
\label{lemma:composition_subtree}
    Let $\alpha, \beta$ be sentential forms, $i,j \in \N$.
    The tree $\FT{\alpha}{i};\FT{\beta}{j}$ is a prefix of $\FT{\alpha \beta}{i+j}$.
    %, i.e. it can be obtained by removing a set of nodes from $\FT{\alpha \beta}{i+j}$.
\end{lemma}

\begin{proof}
    By the definition of relational composition, a branch of $\FT{\alpha}{i};\FT{\beta}{j}$ can be decomposed into a branch of $\FT{\alpha}{i}$ of length at most $i$ and a branch of $\FT{\beta}{j}$ of length at most $j$.  
    The total length of the branch is at most $i+j$, which means it will also occur in $\FT{\alpha \beta}{i+j}$, unless other nodes that were removed in $\FT{\alpha \beta}{i+j}$ trigger some node on the branch to be marked. 
    Towards a contradiction, take the deepest node on the branch that gets marked triggered by a set of nodes of depth greater than $i+j$ being marked. 
    Since it got marked although it has one successor that will not be marked, namely the one on the branch, it has to be a conjunction. 
    Each triggering node is either in the part of the tree corresponding to the derivation process of $\alpha$ or in the part of the tree corresponding to the derivation process of $\beta$, i.e.\ there is a corresponding node in $\FT{\alpha}{}$ or $\FT{\beta}{}$. 
    In the first case, the corresponding node has depth $ i+j +1 > i$, so it will be removed in $\FT{\alpha}{i}$. 
    In the second case, it either occurs on a branch of $\FT{\beta}{}$ that was appended to a branch of $\FT{\alpha}{}$ of length greater than $i$, which will be removed in $\FT{\alpha}{i}$, or it occurs on a branch of length greater than $j$ in $\FT{\beta}{}$, which will be removed in $\FT{\beta}{j}$.
    In any case, all triggering nodes will also be removed to obtain $\FT{\alpha}{i};\FT{\beta}{j}$, which would lead to the removal of the node on the branch in $\FT{\alpha}{i};\FT{\beta}{j}$.
    This is a contradiction to the fact that the branch appears there.
        
    To obtain $\FT{\alpha}{i};\FT{\beta}{j}$ from $\FT{\alpha \beta}{i+j}$ by removing nodes, note that any branch of $\FT{\alpha \beta}{i+j}$ decomposes into a branch of $\FT{\alpha}{i+j}$ and a branch of $\FT{\beta}{i+j}$ so that their total length is at most $i+j$. 
    If the length of the $\alpha$-part is longer than $i$, remove the first node of depth strictly greater than $i$ (this will also remove the $\beta$-part potentially appended to it).
    If the length of the $\beta$-part is longer than $j$, remove the first node of relative depth strictly greater than $j$.
\end{proof}
Combined with Lemma \ref{lemma:remove_superset_implies}, we obtain 
$\FT{\alpha}{i};\FT{\beta}{j} \lleq \FT{\alpha \beta}{i+j}$.

We can now prove the fundamental correspondence between the $i^{\text{th}}$ Kleene approximant and the formula tree cut at level $i$. 
We do not get a precise result like $\sol{\alpha}^i \lsim \FT{\alpha}{f(i)}$. 
This is due to the fact that exploring the tree for one more level and doing one step of Kleene iteration behave differently.
Exploring the tree for one more level will consider one more derivation step in each branch. 
This derivation step will be applied only to the leftmost non-terminal of the sentential form forming the last node in a branch of length $i$. 
Doing one step of Kleene iteration will replace $\sol{X}^i$ by $\sol{X}^{i+1}$ for every non-terminal, which means applying one derivation to each non-terminal that occurs at level $i$. 
A sentential form obtained from $\alpha$ and represented by the $i^{\text{th}}$ Kleene approximation has at most $|\alpha| |G|^i$ symbols. 
Indeed, in each step, we replace a non-terminal by at most the number of symbols in the largest right-hand side of any rule in $G$. 
We can think the Kleene iteration as exploring the trees for all non-terminals simultaneously and composing them to get the formula for $\alpha$, instead of just exploring the tree for $\alpha$ directly.

\begin{lemma}
\label{lemma:ith_sol_vs_ith_tree}
    $\FT{\alpha}{i}\lleq\sol{\alpha}^i\lleq\FT{\alpha}{ |\alpha| |G|^i}$.
\end{lemma}

\begin{proof}
    We prove the statement using induction on $i$.
    
    Let $i = 0$. We consider two cases.
    If $\alpha$ is a terminal word, $\alpha = w \in T^*$, then the tree of plays actually only consists of the root node labeled by $\tbox{w}$.
    Therefore, we have
    \[
        \FT{w}{|w| |G|^0}
        = \FT{w}{|w|}
        = \FT_{w}^0
        = \FT_{w}
        = \tbox{w}
        \lsim \sol{w}^0\ .
    \]
    If $\alpha$ contains a non-terminal, then $\FT_{\alpha}$ has height at least $1$.
    This means the root node is a conjunction or disjunction, which will be deleted since all its successors of depth larger than $0$ will be marked.
    The empty tree remains, which we identify with $\ifalse$.
    Furthermore, \mbox{$\sol{\alpha}^0 = \ifalse$}, since we initialize the solution with $\ifalse$ for every non-terminal.
    As the implication $\FT{\alpha}{i} \lleq \FT{\alpha}{|\alpha| |G|^i}$ always holds, we obtain
    $
        \sol{\alpha}
        \lsim \FT{\alpha}{0}
        \lleq \FT{\alpha}{|\alpha| |G|^0}
    $.

    Let $i > 0$.
    We prove the statement by an inner induction on the structure of $\alpha$.
    Assume $\alpha = X$ is a single non-terminal.
    (If $\alpha = \varepsilon$ or $\alpha = a \in T$, the proof for $i = 0$ carries over.)
    Assume $X$ is owned by prover.
    By the induction hypothesis,
    \[
        \FT{\eta}{i-1}
        \lleq
        \sol{\eta}^{i-1}
        \lleq
        \FT{\eta}{|\eta| |G|^{i-1}}
    \]
    for any rule $X \to \eta$.
    By the monotonicity of conjunction,
    \[
        \bigwedge_{X \to \eta}
        \FT{\eta}{i-1}
        \lleq
        \bigwedge_{X \to \eta} \sol{\eta}^{i-1}
        \lleq
        \bigwedge_{X \to \eta}
        \FT{\eta}{|\eta| |G|^{i-1}}\ .
    \]
    We have 
    \mbox{$\sol{X}^i = \bigwedge_{X \to \eta} \sol{\eta}^{i-1}$}\ ,
    and since formula trees are created by applying grammar rules, we know that
    $
        \bigwedge_{X \to \eta} \FT{\eta}{i-1}
        =
        \FT{X}{i}
    $.
    Furthermore, 
    \[
        \bigwedge_{X \to \eta}
        \FT{\eta}{|\eta| |G|^{i-1}}
        =
        \FT{X}{|\eta| |G|^{i-1} +1}
        \lleq
        \FT_{X}^{|G|^i}\ .
    \]
    The implication is due to $|\eta| \leq |G| - 1$ and Lemma~\ref{lemma:remove_superset_implies}. 
    Altogether, 
    $
        \FT{X}{i}
        \lleq
        \sol{X}^i
        \lleq
        \FT{X}{|G|^i}
    $.     
    If $X$ is owned by refuter, the proof is similar (disjunction is also monotonic).
        
    Let $\alpha = x.\gamma$.
    We may assume by the inner induction that
    $
        \FT{\gamma}{i}
        \lleq
        \sol{\gamma}^{i}
        \lleq
        \FT{\gamma}{|\gamma| |G|^i }
    $
    holds.
    As in the base case, we know
    $
        \FT{x}{i}
        \lleq
        \sol{x}^{i}
        \lleq
        \FT{x}{|G|^i}
    $.
    By an argumentation analogous to the one used in the proof  Lemma~\ref{lemma:composition_subtree}, $\FT{x \gamma}{i}$ is a prefix of
    $\FT{x}{i};\FT{\gamma}{i}$, so
    \[
        \FT{x \gamma}{i}
        \lleq \FT{x}{i} ; \FT{\gamma}{i}
        \lleq \sol{x}^i ; \sol{\gamma}^i
        = \sol{x\gamma}^i\ .
    \]
    The first implication is Lemma~\ref{lemma:remove_superset_implies}, the second is by monotonicity of composition. 
    The equality holds by definition.

    Similarly,
    $
        \sol{x}^i ; \sol{\gamma}^i
        \lleq
        \FT{x}{|G|^{i}} ; \FT{\gamma}{|\gamma| |G|^i}
    $ holds by monotonicity,
    and $\FT{x}{|G|^{i}} ; \FT{\gamma}{|\gamma| |G|^i}$
    is a prefix of
    $\FT{x \gamma}{(|\gamma| + 1) |G|^{i}}$ by Lemma~\ref{lemma:composition_subtree},
    so we conclude
    $
        \sol{x\gamma}^i
        \lleq
        \FT{x\gamma}{(|\gamma| + 1) |G|^{i}}
        =
        \FT{x\gamma}{|x \gamma| |G|^{i}}
    $.
\end{proof}
We now lift the correspondence between the $i^{\text{th}}$ Kleene approximant and the tree cut at level~$i$ to a correspondence between the (usually infinite) formula tree and the fixed-point solution. 
The result is an exact characterization of the fixed-point solution.

\begin{theorem}
    $\FT_\alpha \lsim \sol{\alpha}$.
\end{theorem}

\begin{proof}
    Assume the value of $\FT_\alpha$ under $\nu$ is $\itrue$.
    By Lemma~\ref{lemma:index_true_cut}, there is a finite index $i_0$ so that the value of $\FT{\alpha}{i_0}$ is $\itrue$.
    By Lemma~\ref{lemma:ith_sol_vs_ith_tree} and the fact that the  fixed-point solution is implied by any approximant,
    $\FT{\alpha}{{i_0}} 
    \lleq\sol{\alpha}^{i_0}
    \lleq \sol{\alpha}$, 
    so $\sol{\alpha}$ also evaluates to $\itrue$ under $\nu$.
        
    Assume $\sol{\alpha}$ evaluates to $\itrue$ under $\nu$. 
    Note that $\sol{\alpha} = \sol{\alpha}^{i_0}$ for some index $i_0 \in \N$. 
    Again by Lemma~\ref{lemma:ith_sol_vs_ith_tree}, we have $ \sol{\alpha}^{i_0} \lleq \FT{\alpha}{|\alpha||G|^{i_0} }$, so the value of $\FT{\alpha}{|\alpha||G|^{i_0} }$ is $\itrue$. 
    Since $\FT{\alpha}{|\alpha||G|^{i_0} }$ is a prefix of $\FT{\alpha}{}$, the value of $\FT_{\alpha}$ is $\itrue$ by Lemma~\ref{lemma:remove_superset_implies}.
\end{proof}
The theorem yields another method for computing the fixed-point solution by exploring the tree of plays up to a finite level. 
We could use the fact that $\sol{} = \sol{}^{i_0}$, where an upper bound for $i_0$ can be computed using $|G|$ and the number of equivalence classes of formulas over $\Boxes_A$.  
%\sm{Should we make this more precise? Can we refer to the complexity section?}
By Lemma~\ref{lemma:ith_sol_vs_ith_tree}, we obtain
\mbox{$\sol{\alpha} = \sol{\alpha}^{i_0} \lsim \FT{\alpha}{|\alpha| |G|^{i_0}}$}, so we may also explore $\FT{\alpha}{}$ up to \mbox{level $|\alpha| |G|^{i_0}$.}

Note that this will not only require exponentially more iterations to obtain the fixed point in the worst case, it is also impractical because $\FT{\alpha}{i} \lsim \FT{\alpha}{i+1}$ does not necessarily imply $\FT{\alpha}{i} \lsim \FT{\alpha}{}$. Unlike in the Kleene iteration, we cannot conclude to have reached the fixed point as soon as one unfolding does not change the solution. 
Consider the grammar given by $S \to X, S \to a, X \to Y, X \to a, Y \to b$ where refuter owns all non-terminals. 
Let the automaton be so that $\tbox{a} \neq \tbox{b}$.
Then we have $\FT{S}{1} \lsim \FT{S}{2} \lsim \tbox{a}$, but
$\FT{S}{2} \not\lsim \FT{S}{3} \lsim \FT{S}{} \lsim \tbox{a} \vee \tbox{b}$.

% strategies
\section{Winning Regions and Strategy Synthesis}
\label{sec:strategies}

Define the set of sentential forms  
 \[
     \wrinclusion = \Set{ \alpha \in \vartheta }{ \sol{\alpha} \text{ is not rejecting } }
 \]
%$
%    \wrinclusion = \Set{ \alpha \in \vartheta }{ \sol{\alpha} \text{ is not rejecting } }
%$
and denote its complement by
 \[
     \wrnoninclusion = \vartheta \setminus \wrinclusion = \Set{ \alpha \in \vartheta }{ \sol{\alpha} \text{ is rejecting } }
     \ .
 \]
%$
%    \wrnoninclusion = \vartheta \setminus \wrinclusion = \Set{ \alpha \in \vartheta }{ \sol{\alpha} \text{ is rejecting } }.
%$
Our goal is to prove the following result in a constructive way, by synthesizing strategies guided by the fixed-point solution to the system of equations.

\begin{theorem}
\label{thm:determinacy}
    Inclusion games are determined:
     \[
         \vartheta = \wrinclusion \cupdot \wrnoninclusion
         \ ,  
     \]
%$
%        \vartheta = \wrinclusion \cupdot \wrnoninclusion
%         ,  
%$
    where $\wrinclusion$ is the winning region of prover and $\wrnoninclusion$ is the winning region of refuter. 
\end{theorem}
As a consequence, it is \emph{decidable} whether a given sentential form $\alpha$ is winning for a player: Compute the formula $\sol{\alpha}$ and evaluate it under $\nu$ to check whether it is rejecting. 

It has been shown in \cite{Serre03} that for all games on pushdown systems with $\omega$-regular winning conditions, the winning regions are regular.
Indeed, the winning region of the non-inclusion game can be accepted by a deterministic automaton.
The set of equivalence classes of formulas forms its set of control states, the equivalence class of $\id$ is the initial state and rejecting formulas are final states.
If the automaton in state $F$ reads symbol $x \in N \cup T$, it switches to the \mbox{state $F;\sol{x}$.}

Representing sentential forms by formulas is too imprecise to do strategy synthesis. (In fact, the leftmost non-terminal is not even encoded in the formula.)
Since relational composition is associative, we can represent the set of all sentential forms $\alpha = w X \beta$ by a set of triples
$ (\sol{w}, X, \sol{\beta})$,
where $\sol{w}$ and $\sol{\beta}$ are taken from a finite set of formulas (up to logical equivalence) and $X$ is a non-terminal from a finite set.
This finite representation will be sufficient for the strategy synthesis.
Our synthesis operates on normalized \mbox{formulas in CNF}.

\subsection{Conjunctive Normal Form}
\label{sec:CNF}
A formula in CNF is a conjunction of clauses, each clause being a disjunction of boxes.
We use set notation and write clauses as sets of boxes and formulas as sets of clauses. 
The set of \emph{CNF-formulas} over $\Boxes_A$ is thus 
$\CNF_A= \pwrset{ \pwrset { \Boxes_A } }$. 
Identify $\itrue=\set{}$ and $\ifalse = \set{ \set{} }$. 
In this section, all formulas will stem from $\CNF_A$. 

Since our CNF-formulas are negation-free, implication has a simple %set-theoretic 
characterization. 

\begin{lemma}
    \label{lem:leq}
    $F \lleq G$ if and only if there is $j : G \to F$ so that $j(H) \subseteq H$ for all $H\in G$. 
\end{lemma}

\begin{proof}
    The implication from right to left is immediate.
    Assume $F\lleq G$ but there is no map~$j$ as required.
    Then there is some clause $H\in G$ so that for every clause $K\in F$ we find a variable $x_K\in K$ with $x_K \notin H$.
    Consider the assignment $\nu(x_K)=\itrue$ for all $x_K$ and $\nu(y)=\ifalse$ for the remaining variables.
    Then $\nu(F)=\itrue$.
    At the same time, $\nu(G)=\ifalse$ as $\nu(H)=\ifalse$. 
    This contradicts the assumption $F\lleq G$, which means $\nu(F)=\itrue$ implies $\nu(M)=\itrue$ for every assignment $\nu$.
\end{proof}
When computing a disjunction, we have to apply distributivity to obtain a CNF.
\begin{lemma}
    \label{lem:disjunctionconjunction}
    $F \vee G \lsim \Set{ K \cup H }{ K \in F, H \in G}$ and 
    $F \wedge G \lsim F \cup G$.
\end{lemma}
When computing the relational composition $F;G$ of CNF-formulas, we obtain a formula with three alternations between conjunction and disjunction. 
We apply distributivity to normalize $F;G$.  
Lemma~\ref{lem:cnf_seqcomp} gives a closed-form representation of the result. 
To understand the idea, consider the composition of one clause with a CNF. 

\begin{example}
    \newcommand {\ta}
    { \tboxr_a }
    \newcommand {\tb}
    { \tboxr_b }
    \newcommand {\tc}
    { \tboxr_c }
    \newcommand {\td}
    { \tboxr_d }
    Consider
    $
        F;G
        = (\ta \vee \tb);(\tc \wedge \td)
        = (\ta;\tc \wedge \ta;\td) \vee (\tb;\tc \wedge \tb;\td)
    $.
    Distributivity yields
    $
        (\ta;\tc \vee \tb;\tc)
        \wedge (\ta;\tc \vee \tb;\td)
        \wedge ( \ta;\td \vee \tb;\tc )
        \wedge (\ta;\td \vee \tb;\td)
    $.
\tqed
\end{example}
To turn $F;G$ to CNF, we normalize the composition $K;G$ for every clause $K\in F$.
The formula $K;G$ is an alternation of disjunction (not in the example), conjunction, and disjunction.
Distributivity, when applied to the topmost two operations,
selects for every box  $\rho\in K$ a clause $H\in G$ to compose $\rho$ with.
This justifies the following set-theoretic characterization.

\begin{lemma}
\label{lem:cnf_seqcomp}
    $
        F;G
        \lsim
            \bigcup_{K \in F}
            \bigcup_{z : K \to G}
            \big\{ 
               \bigcup_{\tboxr \in K} \tboxr;z(\tboxr)
            \big\}
    $. 
\end{lemma}
For negation-free formulas, the presence of the empty clause characterizes $\ifalse$.
\begin{lemma}
    \label{lemma:falseclause}
    $F\lsim \ifalse$ if and only if $\set{}$ is a clause of $F$.
\end{lemma}

\begin{proof}
    Note that $\ifalse \lleq F$ holds for any formula.
    In the following, we use the characterization of implication given in Lemma \ref{lem:leq}.
    Assume that $F$ contains the empty clause. We define $j : \ifalse \to F$ by mapping the clause $\set{}$ of $\ifalse = \set{ \set{} }$ to the empty clause in $F$.
    Assume that $F \lsim \ifalse$. In particular, the empty clause in $\ifalse$ embeds a clause of $F$. The embedded clause has to be the empty clause.
\end{proof}

\subsection{Strategy for Prover}
\label{sec:inclusion_prover}

Prover wins on infinite plays, and therefore does not have to care about termination. 
This yields a simple positional winning strategy.

\begin{theorem}
\label{thm:winning_strat_inclusion}
    The strategy $\wstratinclusion$ that applies a rule such that the formula for the resulting position is not rejecting (if possible) is a winning strategy for prover for the inclusion game from all positions in $ \wrinclusion$.
\end{theorem}
Strategy $\wstratinclusion$ is \emph{uniform} and \emph{positional}. 
Uniform means it is winning from all position in $\wrinclusion$. 
The strategy is positional in that it only needs to consider the current position in order to make a choice.  
Moreover, we can precompute $\sol{\eta}$ for all $\eta$ occurring as the right-hand side of a rule. 
Together with the representation of sentential forms as triples 
$(\sol{w}, X, \sol{\beta})$,  
this will be sufficient to decide on prover's next move.
Hence, the strategy is not only positional (but depending on an infinite set of positions), it can even be 
\mbox{\emph{finitely represented}.}
One can also understand it as the winning strategy for a finite game where the triples $(\set{ \set{ \tbox{w}} }, X, \sol{\beta})$ form the nodes and the production rules induce the edges. 

For the proof, we show that whenever we are in a non-rejecting position and it is prover's turn, there is a move to a non-rejecting position. % (Lemma~\ref{lemma:inclusion_strat}(1)).
Hence, if we start from $\wrinclusion$, the condition on the existence of a move (stated in the theorem as \emph{if possible}) does not apply. 
Refuter can only move to positions with non-rejecting formulas. 

\begin{lemma}
    \label{lemma:inclusion_strat}
    Let $\alpha = wX\beta \in \vartheta$ with $\sol{\alpha}$ not rejecting.\\
    (1) If $X \in N_\square$, there is $X \to \eta$ so that $\sol{w \eta \beta}$ is not rejecting. \\    
    (2) If $X \in N_\roundss$, then $\sol{w \eta \beta}$ is not rejecting for all $X \to \eta$.
\end{lemma}
\begin{proof}
    To prove Lemma~\ref{lemma:inclusion_strat}, note that
    $\sol{w X \beta} = \sol{w};\sol{X};\sol{\beta}$.
    Since $w$ is a terminal word, we have
    $\sol{w} = \set{ \set{ \tbox{w} } }$. 
    Hence, a clause of $\sol{w};\sol{X}$ will be of type $\tbox{w};K$, where $K$ is a clause\mbox{ of $\sol{X}$.}
    
    \begin{enumerate}[(1)]
        \item
        If $\sol{\alpha}$ is not rejecting, it contains a clause $K'$ without a rejecting box.  
        By Lemma~\ref{lem:cnf_seqcomp}, $K'$ is defined by a clause $K$ of $\sol{wX}$ and a mapping $z$ from the boxes of $K$ to the clauses of $\sol{\beta}$. 
        Since $X$ is owned by prover, we have
        $\sol{X} = \bigwedge_{ X \to \eta} \sol{\eta}$. 
        The conjunction is a union of the sets of clauses, Lemma~\ref{lem:disjunctionconjunction}. 
        Hence, clause $K$ also occurs in $\sol{w\eta}$ for some $\eta$. 
        By choosing the same mapping $z$, we get that $K'$ is also a clause in $\sol{w \eta \beta}$.   
        \item
        We prove (2).
        If $\sol{\alpha}$ is not rejecting, it contains a clause $K'$ without a rejecting box. 
        By Lemma~\ref{lem:cnf_seqcomp}, $K'$ is determined by a clause $\tbox{w};K$ of $\sol{wX}$ and a map \mbox{$z : \tbox{w};K \to \sol{\beta}$} mapping boxes to clauses. 
        Since $X$ is owned by refuter, we have $\sol{X} = \bigvee_{ X \to \eta}\sol{\eta}$.
        By the characterization of disjunction (Lemma~\ref{lem:disjunctionconjunction}), there is a representation of
        $\tbox{w};K$ as $\bigcup_{X \to \eta} \tbox{w};K_{\eta}$
        with $K_{\eta} \in \sol{\eta}$ for all $X \to \eta$.
        
        Consider a rule $X \to \eta$.
        Let $K''$ be the clause of $\sol{w \eta \beta}$ determined by the clause $\tbox{w};K_{\eta}$ and the map $z$ (nproperly restricted).   
        This clause is not rejecting.
        If $K''$ contained a rejecting box, $K'$ would also contain this box since $\tbox{w};K_{\eta} \subseteq \tbox{w};K$. 
        A contradiction.
    \end{enumerate}
\end{proof}
\begin{proof}[Proof of Theorem~\ref{thm:winning_strat_inclusion}]
    For $w\in T^*$, we have $\sol{w} = \snglt{ \tbox{w} }$. In particular, $w \in \lang{A}$ if and only if $\sol{w}$ is not rejecting.  
    This shows that $\lang{A} \subseteq \wrinclusion$ and $\overline{\lang{A}} \cap \wrinclusion = \emptyset$.
    
    We show that all positions occurring in a play conform to $\wstratinclusion$ and starting in a position from $\wrinclusion$ remain in $\wrinclusion$. 
    This proves the claim since we either obtain an infinite play or we end up in a position in $\lang{A}$. 
    In both cases, the inclusion winning condition is satisfied.
    Technically, whenever prover owns the leftmost non-terminal, the strategy will choose a rule such that the new position still has a non-rejecting formula. 
    By Lemma~\ref{lemma:inclusion_strat}(1) below, this is possible. 
    Whenever refuter owns the leftmost non-terminal, she can only choose a rule such that new position still has a non-rejecting formula by Lemma~\ref{lemma:inclusion_strat}(2).
\end{proof}

\subsection{Non-Inclusion (for Refuter)}
\label{sec:strategy_refuter}
A CNF-formula is rejecting iff for each clause chosen by prover, refuter can select a rejecting box in this clause.  
We formalize the selection process using the notion of choice functions. 
A \emph{choice function} on $F\in \CNF_A$ is a function
\mbox{$c : F \to \Boxes_A$}
selecting a box from each clause, $c(K) \in K$ for all $K \in F$.
We show that there is a strategy for refuter to derive a terminal word from one of the chosen boxes.  
In particular, the strategy will only generate finite plays. 
Note that a choice function can only exist if $F$ does not contain the empty clause.  
Otherwise, the formula is equivalent to $\ifalse$ (Lemma~\ref{lemma:falseclause}), and refuter cannot enforce termination of the derivation process. 

We show that by appropriately selecting the moves of refuter, we can refine the choice function along each play, independent on the choices of prover. 
Given a choice function $c$ on a CNF-formula $F$, a choice function $c'$ on $G$ \emph{refines} $c$ if $\Set{ c'(H) }{ H \in G } \subseteq \ChosenBoxes$, denoted by $c'(G) \subseteq c(F)$.
Given equivalent CNF-formulas, a choice function on the one can be refined to a choice function on the other formula. 
Hence, we can deal with representative formulas in the following proofs.

\begin{lemma}
\label{lemma:choice_fct}
    Consider $F \lleq G$.
    For any choice function $c$ on $F$, there is a choice function $c'$ on $G$ that refines it.
\end{lemma}

\begin{proof}
    By Lemma~\ref{lem:leq}, any clause $H$ of $G$ embeds a clause $j(H)$ of $F$. We can define $c'(H)$ as $c(j(H))$ to get a choice function with $c'(G)\subseteq c(F)$.
\end{proof}
To construct the strategy, we consider formulas obtained from Kleene approximants.
Define a \emph{sequence of levels} $\lvl$ associated to a sentential form $\alpha$ to be a sequence of natural numbers of the same length as $\alpha$.  
The formula $\sol{\alpha}^{\lvl}$ corresponding to $\alpha$ and $\lvl$ is defined by $\sol{a}^{i} = \snglt{\tbox{a}}$ for all $a\in T\cup\set{\varepsilon}$, $\sol{X}^{i}$ the solution to $X$ from the $i^{\text{th}}$ Kleene iteration,  
and $\sol{\alpha.\beta}^{\lvl.\lvl'} = \sol{\alpha}^{\lvl} ; \sol{\beta}^{\lvl'}$.
A choice function for $\alpha$ and $\lvl$ is a choice function on $\sol{\alpha}^{\lvl}$.
Note that $\sol{a}^{i}$ is independent of $i$ for terminals $a$.
Moreover, there is an $i_0$ so that $\sol{X}^{i_0} = \sol{X}$ for all non-terminals $X$.
This means a choice function on $\sol{\alpha}$ can be understood as a choice function on $\sol{\alpha}^{i_0}$.
Here, we use a single number $i_0$ to represent a sequence $\lvl=i_0 \ldots i_0$ of the appropriate length.

By definition, $\sol{X}^0$ is $\ifalse$ for all non-terminals, and $\ifalse$ propagates through relational composition by definition. 
We combine this observation with the fact that choice functions do not exist on formulas that are equivalent to $\ifalse$.

\begin{lemma}
\label{lem:refuterbasecase}
    If there is a choice function for $\alpha$ and $\lvl$, then $\lvl$ does not assign zero to any non-terminal $X$ in $\alpha$.
\end{lemma}
The lemma has an important consequence.
Consider a sentential form $\alpha$ with an associated sequence $\lvl \in 0^*$ and a choice function $c$ for $\alpha$ and $\lvl$.
Then $\alpha$ has to be a terminal word, $\alpha = w \in T^*$, $\sol{\alpha}^{\lvl} = \snglt{\tbox{w}}$, and the choice function has to select $\tbox{w}$. 
In particular, $w$ itself forms a maximal play from this position on, and indeed the play ends in a word whose box is contained in the image of the choice function. 

Consider now $\alpha =wX\beta$ and $\lvl$ an associated sequence of levels.
Assume $\lvl$ assigns a positive value to all non-terminals.  
Let $j$ be the position of $X$ in $\alpha$ and let $i = \lvl_j$ be the corresponding entry of $\lvl$. 
We split $\lvl = \lvl' . i . \lvl''$ into the prefix for $w$, the entry $i$ for $X$, and the suffix for $\beta$.
For each rule $X \to \eta$, we define
%\[
%    \lvl_\eta = \lvl' . (i-1) \dots (i-1) . \lvl''
%\]
$
    \lvl_\eta = \lvl' . (i-1) \dots (i-1) . \lvl''
$
to be the sequence associated  to $w \eta \beta$. 
It coincides with $\lvl$ on $w$ and $\beta$ and has entry $i-1$ for all symbols in $\eta$.
Note that for a terminal word, the formula is independent of the associated level, so we have
$\sol{wX}^{\lvl'.i} = \sol{wX}^{i}$
and
$\sol{w \eta}^{\lvl'. (i-1) \ldots (i-1)} = \sol{w \eta}^{i-1}$.

We show that we can (1) always refine a choice function $c$ on $\sol{\alpha}^{\lvl}$ along the moves of prover and (2) whenever it is refuter's turn, pick a specific move to refine $c$.

\begin{lemma}
\label{lemma:strategy_refinement}
    Let $c$ be a choice function for $\alpha=wX\beta$ and $\lvl$.\\
    (1) If $X\in N_{\square}$, for all $X \to \eta$ there is a choice function $c_\eta$ for $w \eta \beta$ and $\lvl_\eta$ that refines $c$.
    (2) If $X\in N_{\roundss}$, there is $X \to \eta$ and a choice function $c_{\eta}$ for $w \eta \beta$ and $\lvl_\eta$ that \mbox{refines $c$.}
\end{lemma}

\begin{proof}
    \mbox{}\\
    \vspace*{-0.6cm}
    \begin{enumerate}[(1)]
        \item
            Let $F = \sol{\alpha}^{\lvl}$ and $F_\eta = \sol{w\eta\beta}^{\lvl_\eta}$.
            By Lemma~\ref{lem:cnf_seqcomp}, the clauses of $F$ are given by a clause $\tbox{w};K$ of $\sol{wX}^{\lvl'.i}=\sol{wX}^{i}$ and a function mapping the boxes in this clause to
            $\sol{\beta}^{\lvl''}$. 
            Similarly, the clauses of $F_\eta$ are given by a clause of $\sol{w \eta}^{i-1}$ and a mapping from the boxes to $\sol{\beta}^{\lvl''}$.
            We have $\sol{X}^i = \bigwedge_{X \to \eta} \sol{\eta}^{i-1}$.
            Since the conjunction corresponds to a union of the clause sets, Lemma~\ref{lem:disjunctionconjunction}, every clause of $\sol{w\eta}^{i-1}$ is already a clause of $\sol{wX}^{i}$. 
            Hence, the clauses of $F_\eta$ form a subset of the clauses of $F$.
            Since $c$ selects a box from every clause of $F$, we can define the refinement $c_\eta$ on $F_\eta$ by restricting $c$.
        \item
           We show that there is a rule $X \to \eta$ and a choice function $c_\eta$ on $\sol{w\eta \beta}^{\lvl_\eta}$ refining $c$.
           Towards a contradiction, assume this is not the case.
           Then for each rule $X \to \eta$, there is at least one clause $K_\eta''$ of $\sol{w\eta \beta}^{\lvl_\eta}$ that does not contain a box in the image of $c$.
           By Lemma~\ref{lem:cnf_seqcomp}, this clause is defined by a clause $\tbox{w};K_\eta'$ of $\sol{w\eta}^{i-1}$ and a function $z_\eta$ mapping the boxes from this clause to $\sol{\beta}^{\lvl''}$.
           
           We have $\sol{X}^i = \bigvee_{X \to \eta} \sol{\eta}^{i-1}$.
           A clause of $\sol{wX}^i$ is thus (Lemma~\ref{lem:disjunctionconjunction}) of the form 
           \[
               K=\tbox{w};( \bigcup_{X \to \eta} K_\eta)=\bigcup_{X \to \eta} \tbox{w};K_\eta
               \ ,
            \]
            where each $K_\eta$ is a clause of $\sol{\eta}^{i-1}$.
            We construct the clause $K' = \tbox{w}; ( \bigcup_{X \to \eta} K_\eta')$ of $\sol{wX}^i$ using the $K_\eta'$ from above.
            On this clause, we define the map $z' = \bigcup_{X \to \eta} z_\eta$ that takes a box $\tbox{w};\tboxr \in \tbox{w};K_\eta'$ and returns $z_\eta(\tbox{w};\tboxr)$. 
            (If a box $\tbox{w};\tboxr$ is contained in $\tbox{w};K_\eta'$ for several $\eta$, pick an arbitrary $\eta$ among these.)
            By Lemma~\ref{lem:cnf_seqcomp}, $K'$ and $z'$ define a clause of $\sol{\alpha}^{\lvl}$. 
            The choice function $c$ selects a box $\tbox{w};\tboxr;\tboxt$ out of this clause, where there is a rule $X \to \eta$ such that $\tboxr \in K_\eta'$ and
            $\tboxt \in z' (\tbox{w};\tboxr) = z_\eta(\tbox{w};\tboxr)$.
            This box is also contained in $K_\eta''$.  
            A contradiction to the assumption that no box from $K_\eta''$ is in the image of $c$.
    \end{enumerate}
    \vspace*{-0.4cm}
\end{proof}
\noindent Notice that the sequence $\lvl_{\eta}$ is smaller than $\lvl$ in the following ordering
$\natlt$ on $\N^*$. 
Given $v, w \in \N^*$, we define $v \natlt w$ if there are decompositions $v = x y z$ and $w = x i z$  so that $i > 0$ is a positive number and $y \in \N^*$ is a sequence of numbers that are all strictly smaller than~$i$.
Note that requiring $i$ to be positive will prevent the sequence $xz$ from being smaller than $x0z$, since we are not allowed to replace zeros by $\varepsilon$.

The next lemma states that $\natlt$ is well founded.
Consequently, the number of derivations $w X \beta\Rightarrow w \eta \beta$ following the strategy that refines an initial choice function will be finite.

\begin{lemma}
\label{lem:wellfounded}
    $\natlt$ on $\N^*$ is  well founded with minimal elements $0^*$.
\end{lemma}

\begin{proof}
    Note that any element of $\N^*$ containing a non-zero entry is certainly not minimal, since we can obtain a smaller element by replacing any non-zero entry by $\varepsilon$. Any element of the form $0^*$ is minimal, since there is no $i$ as required by the definition of $\natlt$.
    
    Assume $v_0 \natgt v_1 \natgt \dots$ is an infinite descending chain. Let $b$ be the maximal entry of $v_0$, i.e.\ $b = \max_{j = 1, ..., |v_0|} v_0$, and note that no $v_l$ with $l \in \N$ can contain an entry larger than $b$ by the definition of $\natlt$.
    Therefore, we may map each $v_l$ to its Parikh image $\psi(v_l) \in \N^{b + 1}$, the vector such that $\psi(v_l)_j$ (for $j \in \zeroto{b}$) is the number of entries equal to $j$ in $v_l$.
    
    Now note that we have $\psi(v_i) > \psi(v_{i+1})$ with respect to the lexicographic ordering on $\N^{b + 1}$.
    Hence, the chain
    $\psi(v_0) > \psi(v_1) > \dots$
    is an infinite descending chain, which cannot exist since the lexicographic ordering is known to be well-founded.
\end{proof}
Lemma~\ref{lem:wellfounded} is used in the main technical result of this section. 
Proposition~\ref{prop:choice_fct} in particular says that all maximal plays that conform to $s_{\alpha,c}$ are finite.
If $\sol{\alpha}$ is rejecting, there is a choice function on $\sol{\alpha}$ that only selects rejecting boxes. The desired theorem is then immediate.   
\begin{proposition}
\label{prop:choice_fct}
    Let $c$ be a choice function on $\sol{\alpha}$.
    There is a strategy $s_{\alpha,c}$ such that all maximal plays starting in $\alpha$ that conform to $s_{\alpha,c}$ end in a terminal word $w$ with
    $\tbox{w} \in c(\sol{\alpha})$.
\end{proposition}

\begin{proof}
    We show the following stronger claim: 
    Given any triple consisting of a sentential form $\alpha$, an associated sequence of levels $\lvl$, and a choice function $c$ for $\alpha$ and $\lvl$, there is a strategy $s_{\alpha,c}$ such that all maximal plays conform to it and starting in $\alpha$ end in a terminal word $w$ with $\tbox{w} \in \Set{ c(K) }{ K \in \sol{\alpha} }$\ . 
    This proves the proposition by choosing $\alpha$ and $c$ as given and $\lvl = i_0 ... i_0$, where $i_0 \in \N$ is a number such that $\sol{} = \sol{}^{i_0}$. 
    
    To show the claim, note that $\natlt$ on $\N^*$ is well founded and the minimal elements are exactly $0^*$ by Lemma~\ref{lem:wellfounded}, and $\lvl_\eta \natlt \lvl$. 
    This means we can combine Lemma~\ref{lem:refuterbasecase} and  Lemma~\ref{lemma:strategy_refinement} (for the step case) into a Noetherian induction. 
    The latter lemma does not state that $\lvl_{\eta}$ assigns a positive value to each non-terminal, which was a requirement on $\lvl$. 
    This follows from Lemma~\ref{lem:refuterbasecase} and the fact that $c_{\eta}$ is a choice function. 
    The strategy $s_{\alpha,c}$ for refuter always selects the rule that affords a refinement of the initial choice function $c$.   
\end{proof}

\begin{theorem}
\label{thm:winning_strat_noninclusion}
    Let $\alpha \in \wrnoninclusion$ and let $c$ select a rejecting box in each clause of $\sol{\alpha}$. 
    Then $s_{\alpha,c}$ is a winning strategy for refuter for the non-inclusion game played from~$\alpha$.
\end{theorem}
\paragraph*{Implementing Winning Strategies}
The strategy $s_{\alpha,c}$ from Proposition~\ref{prop:choice_fct} is not positional.
However, it can be implemented as a \emph{strategy with finite memory}. 
We initialize it with a choice function on $\sol{\alpha}$ and keep track of the current refinement of this function in each play. 
Observe that it suffices to store the image of the current choice function, which is a set of boxes, and the number of boxes is finite.
To further optimize the implementation, instead of storing all Kleene approximants we annotate every box in the fixed point by the iteration step in which it entered the solution.
Rather than selecting an arbitrary rejecting box from each clause to initialize the choice function, one should then choose the rejecting box which entered the solution the earliest.

The strategy can be implemented without levels but at the expense of an enumeration. 
Whenever prover makes a move, the strategy computes the refinement of the choice function.
Given a position $wX\beta$ owned by refuter, it checks in increasing order, for all $j = 0, 1, \ldots$ whether there is a rule $X \to \eta$ such that there is a refinement of the current choice function on $\tbox{w};\sol{\eta}^j; \sol{\beta}$. 
Lemma~\ref{lemma:strategy_refinement}(2) guarantees the existence of such a rule.
Moreover, in terms of levels, $j$ will be smaller than the current level of $X$, say $i$. 
This shows termination of the overall procedure.  
Another advantage of this enumeration of Kleene approximants is that we may find a refinement of the choice function with an iteration number $j<i-1$. 

One can also use a pushdown automaton to implement the strategy. Its stack will always have one entry for each non-terminal of the sentential form currently under consideration, storing the symbol, the level, and the formula for the corresponding suffix of the sentential form.
The automaton iterates over the grammar rules and uses the stored level for the current non-terminal and the formula for the suffix to determine which rule to pick.
The implementation using bounded space requires linear time (in the size of the current sentential form) to select rules.
In contrast to this, the pushdown strategy needs to be initialized once for the initial position and then can determine rules in constant time.

As soon as we know that a winning strategy for refuter exists (i.e.\ by evaluating the formula for the initial position), we can find one by a breadth-first search in the tree formed by all plays. 
The winning strategy will only generate finite plays and the tree has finite out-degree. 
So, by K\"onig's lemma, the tree formed by all plays conform to this strategy has to be finite.  
This allows us to obtain a \emph{positional strategy}.

\begin{example}
    In the running example, 
    formula $\sol{Y} = \set{ \set{ \tbox{b}  } }$ is rejecting.   
    In fact, refuter can win the non-inclusion game played from $Y$. 
    The initial choice function on $\sol{Y}$ has to be $c ( \set{ \tbox{b}  }) = \tbox{b}$. 
    In the first step, prover has no alternative but $Y \to bX$.
    Position $bX$ has the formula 
    $\sol{bX} 
    = \set{ \set{ \tbox{b}  } }; \set{ \set{ \id, \tbox{ab} } }
    = \set{ \set{ \tbox{b}, \tbox{bab} } }
    = \set{ \set{ \tbox{b} } }$.
    Pick the same choice function as before.  
    The rule $X \to \varepsilon$ causes $\id$ to enter $\sol{X}$ in the first Kleene step. 
    This causes $\tbox{b}$ to enter $\sol{bX}$ also in the first step.
    Indeed, by choosing $X \to \varepsilon$ refuter wins non-inclusion.
    \tqed
\end{example}

% complexity
\section{Complexity}
\label{sec:complexity}

We show that deciding whether refuter has a winning strategy for non-inclusion from a given position is a \textsf{2EXPTIME}-complete problem.
Moreover, the algorithm presented in the previous sections
%(even without the improvements discussed in Section \ref{sec:antichains})
achieves this optimal time complexity.

\subsection{Hardness}

We prove that deciding the non-inclusion game is $\mathsf{2EXPTIME}$-hard.
Our proof of the lower bound follows the proof of the analogue result for the games considered in~\cite{Muscholl2005}.

\begin{theorem}
\label{thm:hardness}
    Given a non-inclusion game and an initial position, deciding whether refuter has a winning strategy from the specified position is $\mathsf{2EXPTIME}$-hard.
\end{theorem}

\begin{proof}[Proof of Theorem \ref{thm:hardness}]
    Assume an alternating Turing machine $M$ with exponential space bound~-~say $2^{n^c}$~-~and an input $w$ is given.
    We construct a polynomial-sized game such that refuter has a winning strategy if and only if $M$ accepts $w$.
    
    A similar reduction has been given in \cite{Muscholl2005}.
    There, the authors consider left-to-right games which work as follows: The first player picks a position in the current sentential form and the second player replaces the non-terminal at this position using a rule of his choice.
    Furthermore, whenever the player picking the positions skips a non-terminal, she is not allowed to choose it later.
    We cannot simply reduce the games considered in \cite{Muscholl2005} to our setting, since in left-to-right games, the player who needs to have the finite winning strategy has to reach the regular target language instead of avoiding it.
    
    Our approach is to use the grammar to generate a sequence of configurations of the Turing machine.
    Afterwards, the non-deterministic automaton detects whether the sequence forms a invalid or non-accepting computation.
    (This is similar to the proof that universality of NFAs is $\mathsf{PSPACE}$-hard \cite{AhoHopcroftUllman}.)
    The game aspect allows to us generate a computation of an alternating Turing machine.
    We let refuter choose the transitions originating in existential states and we let prover choose the transitions originating in universal states.
    Since each configuration has exponential size, the polynomial-sized automaton is not able to check whether the head pointer of the Turing machine has been moved in an incorrect way (and therefore, the sequence of configurations is invalid).
    To solve this problem, we divide the derivation process into two parts:
    In the first part, refuter generates a sequence of configurations from right-to-left, only giving control over to prover to pick the transitions originating in universal control states.
    In the second part, prover does a right-to-left pass in which she wants to show that the configuration is invalid (or non-accepting).
    To this end, she places a marker on the encoding of a cell.
    Refuter tries to justify the computation and rebut the objection.
    To this end, she places a marker on the encoding of the cell with the same position in the previous configuration.
    With the help of the markers, the automaton is able to detect violations.
    The automaton will accept the resulting word if and only if the markers prove that it is an invalid computation.
    
    We will use the symbols $0,1,[,],(,) , \langle, \rangle$, and $,$\ .
    Furthermore, each control state $q \in Q_M$, each tape symbol $a \in \Gamma_M$, and each transition $\delta \in \Delta_M$ of the Turing machine is used as a symbol.
    Each of these symbols $x$ will occur in four different versions in the grammar.
    Without any decoration, it is a non non-terminal $x$ owned by prover.
    It may also be a non-terminal $\bar{x}$ owned by refuter or a terminal $x_p$ or $x_r$, where the subscripts indicate which player caused the symbol to be derived.
    Furthermore, there is an initial non-terminal $S$ owned by refuter, several versions of the non-terminal $\#$, and the special terminal symbols $\mathit{START}, \mathit{STOP}, \mathit{ONUM}, \mathit{OTAPE}, \mathit{J}$.
    
    At first, the derivation process will generate a sequence of configuration in a right-to-left fashion.
    This means that production rules of the shape $\# \to \# x$ are applied.
    Note that since we only consider left-derivations, the symbols $x$ that are inserted will not be touched as long as $\#$ has not yet been replaced.
    Each configuration is encoded as a sequence
    \[
    (a_0, \bin(0) ) \
    (a_1, \bin(1) )
    \dots
    (a_{i-1}, \bin(i-1) ) \
    [q, (a_{i}, \bin(i) )]
    %(a_{i+1}, \bin(i+1) )
    \dots
    (a_{2^{n^c}}, \bin(2^{n^c}))
    \ ,
    \]
    where $a_0 ... a_{2^{n^c}}$ is the content of the tape, $q$ is the control state, and $i$ is the position of the head pointer.
    Each cell also contains an encoding of its index on the tape.
    This is crucial for being able to detect invalid head pointer movement later.
    We require that each configuration has length $2^{n^c}$, i.e.\ we will show the trailing blank symbols explicitly.
    Two configurations are separated by a sequence $\langle \delta \rangle$, where $\delta$ should be the transition of the Turing machine that was used to get from the old configuration (on the right) to the new configuration (on the left).
    
    Once the symbol $\#$ has been replaced, the sentential form is a sequence of non-terminals that are all owned by prover.
    Prover now can do a left-to-right pass in which she can place the marker as explained above.

    To be precise, the derivation process will proceed in four phases:
    (1)
    The initial configuration of the machine is generated.
    (2)
    A sequence of configurations of the machine, separated by the transitions that were used, is generated from right to left.
    %This will allow as to do another left-to-right pass using left-derivations after the sequence of configurations has been written.
    (3)
    Afterwards, we need to check that the sequence indeed encodes an accepting branch of a computation of $M$ on $w$.
    A part of this is implemented in the grammar.
    Another part of this is checked by the inclusion in the language of the automaton.
    We elaborate on this below.
    For two of the conditions, we need to use the help of the players, since grammar and automaton are only allowed to be polynomial-sized, and thus cannot process the exponential-sized configurations easily.
    Prover does a left-to-right pass over the sequence in which she can place an objection marker to show why the sequence is not accepting. 
    (4)
    Refuter can then place a justification marker in the next configuration to rebut the the objection.
    
    The automaton checks that the resulting terminal word contains exactly one objection that was not correctly rebutted. This means prover wins if the sequence does not encode an accepting configuration.
    We will now look at each phase in detail.
    
    \begin{enumerate}[(1)]
        \item 
        In the first phase, the initial configuration is generated on the tape.
        During this phase, none of the players has a choice.
        After it has finished, the resulting sentential form is
        \[    
        \# \
        [ q_0, (\$, \bin(0)) ] \
        (w_1, \bin(1))
        \dots
        (w_k, \bin(k)) \ 
        (\blank, \bin(k+1)
        \dots
        (\blank, \bin(2^{n^c})) \ 
        \langle \mathit{START} \rangle
        \ ,     
        \]
        where $\$$ is the marker for the left end of the tape, $q_0$ is the inital state and $w = w_1 ... w_k$.
        Note that it is possible to generate the exponential amount of $(2^{n^c} - k)$ trailing blank symbols with a polynomial-sized grammar.
        \item
        In the second phase, the string is prolonged to the left to represent a branch of the computation tree of $M$ on $w$.
        We give a rough overview of how this can be implemented in the grammar. Later, we will refine this construction.
        
        Initially, refuter has to replace $\#$.
        To do this, refuter has four choices:
        She can
        \mbox{(a)
            write} an arbitrary sequence of symbols $x$, representing control states, tape symbols and the symbols $(, ), [, ], 0, 1$ and $,$ using rules of the shape $\# \to \# x$ ,
        (b)
        write a transition $\delta \in \Delta_\vee$ of the Turing machine originating in an existential state by choosing
        $\# \to \#  \langle \delta \rangle $,
        (c)
        give control over to prover by choosing $\# \to \#_P$.
        Prover can then write a transition $\delta \in \Delta_\wedge$ originating in a universal state and by choosing 
        $\#_P \to \#  \langle \delta \rangle$.
        Or refuter can
        (d)
        stop writing the branch by replacing $\#$ by the sequence $\langle \mathit{STOP} \rangle$, which will end the second phase.
    \end{enumerate}
    The sentential form obtained after Phase (2) is not a correctly encoded sequence of configurations that forms an accepting branch of a computation of $M$ on $w$ if and only if at least one of the following conditions is satisfied.
    
    \begin{enumerate}[(a)]
        \item
        There is a top-level syntax error, i.e.\ the sentential form is not of the shape
        \[
        \langle \mathit{STOP} \rangle
        \left\{
        \{ 0+1+q+a+[+]+(+)+, \}^*
        \
        \langle \delta \rangle
        \right\}^*
        \{ 0+1+q+a+[+]+(+)+,  \}^*
        \langle \mathit{START} \rangle
        \ .
        \]
        Here, we use curly braces for regular expressions to avoid ambiguity. The symbols $q, a, \delta$ denote control states, tape symbols, and transitions respectively.
        \item
        There is a low-level syntax error, i.e.\ one single configuration is not encoded as
        \[
        \left\{ (a, \{0 + 1\}^*) \right\}^*
        [ q, (a, \{0 + 1\}^*) ]
        \left\{ (a, \{0 + 1\}^*) \right\}^*
        \ .
        \]     
        \item
        There is a transition that was picked although it was not applicable.
        This means that control state or the symbol at the position of the head pointer do not match with the preconditions of the transition.
        Since prover can exclusively write transitions that require a universal control state, this will enforce that we let prover choose such transitions.
        \item
        There is a configuration such that its control state is not the state resulting from the transition that was applied before.
        \item
        The leftmost configuration does not contain the accepting control state $q_{accept}$.
        \item
        There is a configuration in which any cell is numbered with a binary string of incorrect length or the first cell is not numbered with $\bin(0)$ or the last cell is not numbered with $\bin(2^{n^c})$.
        \item
        There is a configuration containing two successive cells $(q, v)(q', v')$ such that $v, v'$ encode numbers $i, j$ with $j \neq i+1$.
        \item
        There is a configuration that was modified in an incorrect way (compared to the previous configuration and the chosen transition).
        This means the head pointer was moved inconsistently or a cell content was changed inconsistently.
        The latter might be the case if a cell was changed that was not affected by a transition, or the cell that was previously the position of the head pointer was not changed according to the chosen transition.
    \end{enumerate}
    The Conditions (a) - (e) can be implemented in the grammar.
    We create polynomially-many copies of the symbol $\#$ that keep track of a constant amount of information.
    For example, to handle a part of Condition (b), we can easily guarantee that exactly one control state is written between two transitions by going to a copy $\#^{\mathit{ctrl}}$ after a transition $\langle \delta \rangle $ was inserted.
    This version of the symbols that allows writing a control state (i.e.\ there is a rule $\#^{\mathit{ctrl}} \to \#^{\mathit{trans}}q$ for each $q \in Q$), but does not allow writing a transition.
    As shown in the rule, we go to a copy $\#^{\mathit{trans}}$ after the control states has been inserted.
    Thus version allows writing a transition (i.e.\ there is a rule
    $\#^{\mathit{trans}} \to \#^{\mathit{ctrl}} \langle \delta \rangle$
    ), but does not allow writing a control state.
    To prevent violations of Condition (d), we can also keep track of last control state $q$ and of the symbol $a$ that is stored in the cell which is the position of the head pointer.
    We can implement the transition relation of the Turing machine in the grammar rules to ensure that only transitions can be picked that require the stored control state and tape symbol.
    This means we can only use a rule that inserts $\langle \delta \rangle$ if $\delta$ is a transition of the shape $(q,a) \mapsto (q', a', d)$.
    After the transition has been inserted, we enforce that in the following configuration, the control state coincides with the control state $q'$ resulting from applying the transition, thus ensuring (d) to hold.
    We can also enforce that the rule that ends the second phase by inserting $\langle \mathit{STOP} \rangle$ can only be picked if the previous control state was $q_{accept}$ to handle Condition (e).
    
    Condition (f) can be implemented in the automaton.
    To check that all binary encodings of numbers have the correct length $n^c$, the automaton guesses non-deterministically where a pattern of the form
    $(q,v)$ occurs, were $q$ is a control state and
    $v \in \{0 + 1\}^*$ is a sequence of length not equal to $n^c$.
    To detect this, we create $n^c$ branches of the automaton that detect sequences of length $0, ..., n^c -1$ and one branch that detects sequences of length at least $n^c +1$.
    The index of the first cell of each configuration has the encoding
    $\bin(0) = 0...0$.
    The automaton can check whether a illegal $1$ occurs in the encoding.
    To do this, it guesses non-deterministically where a pattern of the form
    \[
    \rangle (a, \{0 + 1\}^* \, 1 \, \{0 + 1\}^*) 
    \ + \ 
    \rangle [q, (a, \{0 + 1\}^* \, 1 \, \{0 + 1\}^*) ]  
    \]
    occurs, where $a$ is a tape symbol, and $q$ is a control state.
    Note that the symbol $\rangle$ marks the left end of configuration since it belongs to the sequence $\langle \delta \rangle$ encoding a transition (or to $\langle \mathit{STOP} \rangle$ ).
    Analogously, the automaton can verify that the index of the last cell has the encoding $\bin(2^{n^c}) = 1...1$.
    
    To check Conditions (g) and (h), we require the help of the players.
    The Phases (3) and (4) of the derivation process will implement this.
    
    \begin{enumerate}
        \item [(3)]      
        After refuter has decided to end the process by choosing
        $ \# \to \langle \mathit{STOP} \rangle $, we have a sentential form in which all other symbols are non-terminals owned by prover.
        Prover now does a left-to-right pass, in which she can mark a place by an objection symbol to show that the configuration is not valid.
        Formally, she can
        (a) replace a symbol by her terminal version by choosing $x \to x_p$,
        (b) place one of two kinds of objection symbols by choosing a rule $x \to \mathit{ONUM}\ x$ or $x \to \mathit{OTAPE}\ x$, or
        (c) give control over to refuter by choosing $x \to \bar{x}$.
        
        The automaton will later check that the resulting terminal word contains exactly one objection symbol.
        If there is no objection, prover admits that the branch is accepting.
        
        Prover has two types of objection symbols.
        The symbol $\mathit{ONUM}$ can be placed in front of two successive cells
        $(a,v)(a', v')$
        (respectively the variants where one of the cells is the position of the head pointer) to claim that $v$ and $v'$ encode numbers $i$ and $j$ such that $j \neq i+1$.
        It is crucial that $j \neq i+1$ can be accepted for a polynomial-sized NFA for binary encoded numbers $i,j \in \zeroto{2^{n^c}}$.
        
        To do the check, we guess the position $l$ of $v'$ (encoding $j$) that will be the rightmost deviation from the encoding of $i+1$ prior to looking at the two cells.
        We store the position in the control state (which we can do since there are only polynomially-many possibilities, namely $\log (2^{n^c}) = n^c$ ).
        
        When the automaton reads the $l^{\text{th}}$ bit of $v$, it stores the bit in its control states.
        When it then reads the $l^{\text{th}}$ bit of $v'$ it is able to accept if the $v'$ does not encode $i+1$.   
        It is crucial that adding one to a binary number will change the rightmost $0$ to $1$ and flip all following $1$s to $0$.
        The following possibilities arise:
        (a)
        After position $l$, at least one zero follows in $v$.
        In this case, the bit at position $l$ in $v'$ should coincide since the addition can be performed further to the right.
        Check this and accept if the bits are different.
        (b)
        After position $l$, only $1$s follow in $v$.
        In this case, the bit at position $l$ in $v'$ should be different since it
        is affected by the addition.
        Check this and accept if the bits coincide.
        
        The symbol $\mathit{OTAPE}$ can be placed in front of a cell $(a,v)$ to claim that it was modified in an incorrect way. This means it was modified compared to the previous configuration although it was not the position of the head pointer in the previous configuration (and therefore should  have stayed the same) or it was the position of the head pointer previously but was not modified according to the transition.
        
        Since the automaton cannot identify the correct position in the next configuration due to the polynomial bound on its number of states, we will employ the help of refuter to identify the cell of the next configuration that the automaton has to compare against.
        
        \item [(4)]
        The automaton guarantees that after an objection symbols is placed, the control is always given to refuter (by choosing the rule $x \to \bar{x}$ instead of $x \to x_p$). Formally, the automaton rejects the input if a second objection symbol occurs or if the prover-version $x_p$ of any terminal occurs.
        
        Assume prover placed the symbol $\mathit{OTAPE}$ in front of a cell in configuration $c$. Prior to Phase (3), the sentential form had the shape $... c \langle \delta \rangle c' ... $, were $c$ is the configuration that resulted from applying transition $\delta$ to configuration $c'$.
        Refuter can now place a justification symbol $J$ in $c'$ to show that the objection is not valid.
        
        If $\mathit{OTAPE}$ was inserted in front of a cell $(a,v)$ in a configuration $c$, refuter should place $J$ in front of the preceding cell $(b,v')$ of the previous configuration $c'$.
        This means if $v$ and $v'$ encode numbers $i$ and $j$, we have $j = i-1$. (If one of the cells is the position of the head pointer, the symbol should be inserted in front of $[$.)
        
        First, the algorithm will check that refuter marked the correct cell.
        To do this, a procedure analogous to the one explained in Phase (3) can be used to accept if $v'$ is not as desired.             
        Then the automaton can use the marked cell $(a,v)$ in $c$ and the three marked cells
        $(b_j, \bin(j)), (b_i, \bin(i)), (b_{i+1}, \bin(i+1))$
        in $c'$ (respectively the version with the control state annotation somewhere) to verify:  
        (a)
        If $(b_i, \bin(i))$ was not the position of the head pointer in $c'$, then the cell should remain unchanged, i.e.\ $b_i = a$.
        (b)
        If $(b_i, \bin(i))$ was the position of the head pointer in $c'$, then $a$ is the symbol that results from applying the chosen transition.
        (c)
        If $(a, \bin(i))$ is the position of the head pointer in $c$, then the control state was previously placed at $(b_j, \bin(j))$, $(b_i, \bin(i))$, or $(b_{i+1}, \bin(i+1))$, according to the chosen transition.
        
        To check (b) and (c), the automaton enforces that there is exactly one transition $\langle \delta \rangle$ separating the configuration containing $\mathit{OTAPE}$ and the configuration containing $J$.
        It can store the transition $\delta$ in its control state while switching between the configurations.
        
        The cases in which $(b', \bin(i))$ is the first or last cell of a configuration can be treated similarly.
    \end{enumerate}   
    Note that the operations performed by the automaton that are described as if they would happen sequentially here are implemented in parallel.
    This can be done using a constant number of intersections and unions, without changing the size of the automaton being polynomial.
    
    Overall, we obtain that if the ATM $M$ accepts $w$, refuter can find an accepting branch, no matter which transitions are chosen by prover in the universal states. If refuter writes down the correct encoding of the  corresponding finite sequence of configurations, she is able to react to any objection prover might place in Phase (3). The finite automaton will not accept, since it cannot detect a valid objection, so we have derived a word outside its language.
    
    If the ATM $M$ does not accept $w$, prover has a way of selecting transitions originating in universal states such that the resulting branch of the computation tree is not accepting, independent of the choices of refuter in the existential states.
    The construction guarantees that prover is allowed to choose the grammar rule that inserts the transition if the control state in the previous configuration was universal.
    By choosing the grammar rules that insert the transitions leading to a non-accepting branch, prover can ensure that no sentential form encoding a finite sequence of configurations ending in an accepting configuration can be derived.
    If refuter tries to cheat by enforcing an infinite derivation, refuter looses the game by definition. If she cheats by choosing an incorrect encoding (that is not prevented by the grammar), prover can place an objection mark to which refuter can not react.
    The automaton will then detect the valid objection and accept, so we have derived a word in its language.
\end{proof}

\subsection{Membership}

The following algorithm implements the fixed-point iteration discussed in  Section \ref{sec:domain}, executed on formulas in CNF 
(see the Subsection \ref{sec:CNF}).

\begin{algo}\label{Algorithm:FixedPoint}
    Given a non-inclusion game and an initial position $\alpha$, the following algorithm computes whether refuter has a winning strategy from the given position.
    \\
    (1) Set $\sol{X}^0 = \ifalse$ for all $X \in N$. Set $i = 0$.
    \\
    (2) Do until $\sol{X}^{i} \lsim \sol{X}^{i-1}$ for all $X \in N$:
        $i = i+1$; $\sol{}^{i} = f( \sol{}^{i-1})$.
    \\
    (3) Compute $\sol{\alpha}$, and return $\itrue$ iff $\sol{\alpha}$ is rejecting.

    \noindent Here, $f$ is the function combining the right-hand sides of the equations as in Definition \ref{def:game_dfa}.
\end{algo}

\begin{theorem}
\label{thm:membership}
    Given a non-inclusion game and an initial position, Algorithm~\ref{Algorithm:FixedPoint} computes whether refuter has a winning strategy from the given position in time
 %   \[
 %       \bigO{ |G|^2 \cdot 2^{2^{|Q|^{c_1}}} + |\alpha| \cdot 2^{2^{|Q|^{c_2}}}  }
 %   \]
$
        \bigO{ |G|^2 \cdot 2^{2^{|Q|^{c_1}}} + |\alpha| \cdot 2^{2^{|Q|^{c_2}}}  }
$
    for some constants $c_1, c_2 \in \N$.
\end{theorem}

\begin{proof}[Proof of Theorem \ref{thm:membership}]
    We will analyze (1) the number of iterations needed to obtain the fixed-point solution (2) the time consumption per iteration (3) the cost of constructing and evaluating the formula for the given initial position.
    
    Let $k = 2^{|Q|^2}$ be the number of boxes.
    Every clause has at most size $k$ and there are $2^k$ different clauses, so every formula has size at most $k \cdot 2^k$.   
    Computing conjunction and disjunction according to Lemma \ref{lem:disjunctionconjunction} is polynomial in the size of the formulas.
    
    To compute the relational composition according to Lemma \ref{lem:cnf_seqcomp}, we need to iterate over the at most $2^k$ clauses and over the at most $(2^k)^k = 2^{k^2}$ functions mapping boxes to clauses.
    Each clause $K$ and function $z$ determines a clause of the relation composition.
    To obtain its boxes, we need to iterate over the at most $k$ boxes $\tboxr$ of $K$ and compute $\tboxr;z( \tboxr )$.
    To do this, we need to iterate over the at most $k$ boxes $\tboxt$ of $z(\tboxr)$ and compute $\tboxr;\tboxt$.
    This requires that we check for each pair $(q,q'')$ of states whether there is a suitable $q'$ such that $(q,q') \in \tboxr, (q,q'') \in \tboxt$.
    Overall, to compute the relational composition of two formulas, we need
    \[
    2^k \cdot 2^{k^2} \cdot k \cdot k \cdot |Q|^3
    \in
    \bigO{ 2^{2^{|Q|^{c_2}}} }
    \]
    steps, for some constant $c_2 \in \N$.
    
    \begin{enumerate}[(1)]
        \item 
        The length of any chain of strict implications of formulas over a set of $k$ atomic propositions is at most $2^k$.
        To prove this, note that modulo logical equivalence, a formula is uniquely characterized by the set of assignments such that the formula evaluates to $\itrue$ under them. Strict implication between two formulas implies strict inclusion between the sets.
        The statement follows since there are at most $2^k$ different truth assignments.
        
        We can use this to obtain that the number of iterations is bounded by $|N| \cdot 2^k$, since the sequence of intermediary solutions is a chain in the product domain, and the height of the product domain is the height of the base domain multiplied by the number of components.
        \item
        Per iteration, we need to carry out at most $|G|$ conjunctions, disjunctions and relational compositions.
        Per grammar rule, we need to compute at most one conjunction or disjunction, depending on the owner of the non-terminal.
        For each symbol on the right-hand side of a grammar rule, we need to compute at most one relational composition.
        
        Overall, for one iteration, we need
        \[
        |G|
        \cdot
        \left(
        2^{2^{|Q|^{c_2}}}
        + 2^{2^{|Q|^{c_3}}}
        + 2^{2^{|Q|^{c_4}}}
        \right)
        \in
        \bigO
        {
            |G| \cdot 2^{2^{|Q|^{c_1}}}
        }
        \]
        steps, for some constants $c_1, c_3, c_4 \in \N$.
        Here, $2^{2^{|Q|^{c_2}}}$ is the cost of computing the relational composition as calculated earlier and $2^{2^{|Q|^{c_3}}}$ and $2^{2^{|Q|^{c_4}}}$ are rough estimations for computing conjunction and disjunction.
    \end{enumerate}
    Combining (1) and (2) together with the rough estimation $|N| \leq |G|$ yields the first summand
    \[
    \bigO{ |G|^2 \cdot 2^{2^{|Q|^{c_1}}}}
    \ .
    \]
    \begin{enumerate}
        \item [(3)]
        Assume the initial position $\alpha$ has length $|\alpha| = l$.
        It remains to compute $l-1$ relational compositions, which can be done in $(l-1) \cdot 2^{2^{|Q|^{c_2}}}$ steps, and to check whether the resulting formula is rejecting.
        The latter can be done by iterating over all at most $2^k$ clauses, and checking whether one of the at most $k$ boxes in them is rejecting. Checking whether a box $\tboxr$ is rejecting can be done in $|Q|$ time, since we need to check for the absence of all pairs $(q_0, q_f)$ in $\tboxr$, where $q_f \in Q_F$.
        
        Overall, we need
        \[
        (l-1) \cdot 2^{2^{|Q|^{c_1}}} + 2^k \cdot k \cdot |Q|
        \in
        \bigO { l \cdot 2^{2^{|Q|^{c_2}}} }
        \]
        steps.
    \end{enumerate}
    \vspace*{-0.5cm}
\end{proof}
The following corollary is an immediate consequence of Theorem \ref{thm:hardness} and Theorem \ref{thm:membership}.

\begin{corollary}
    Deciding whether refuter has a winning strategy for a given non-inclusion game and an initial position is $\mathsf{2EXPTIME}$-complete.
\end{corollary}
One should note that the running time of the algorithm is only exponential in the size of the automaton. If the automaton is assumed to be fixed, the running time of the algorithm is polynomial in the size of the grammar and in the length $l$ of the initial position. 
Namely, it can be executed in $\bigO{ |G|^2 + l}$ steps.

\subsection{Solving More General Games}

One should also note that the algorithm can solve games on the game arena induced by a grammar not only in the case of the non-inclusion winning condition, but also in a more general setting.
Assume the winning condition is specified by a predicate on the boxes of the automaton (respectively by a predicate on words that is well-defined on the equivalence classes introduced by boxes).
We can lift the definition to obtain a predicate on CNF-formulas over boxes by distributing it over conjunction and disjunction.
We get that a formula satisfies the predicate if and only if there is a choice function picking a box out of each clause that satisfies the predicate.
Initializing Proposition~\ref{prop:choice_fct} with such a choice function will provide a winning strategy that ensures that the game ends in a word such that its box satisfies the predicate after finitely many steps.

If the predicate can be evaluated in doubly exponential time (in the size of the whole input), the time complexity of the algorithm does not change.
In this paper, we mostly considered the $\mathit{reject}$-predicate, that checks for the absence of a transition $(q_0, q_f)$ with $q_f \in Q_F$ in boxes (and therefore is satisfied if the words are not in the language of the automaton).
We can also consider its negation, the $\mathit{accept}$-predicate, that checks for the presence of such a transition (and therefore is satisfied if the words are in the language of the automaton).
Instantiating the algorithm for this predicate yields a procedure to solve a type of game in which refuter wants to obtain a word in the regular target language after finitely many steps.
One can show that solving this problem is also $\mathsf{2EXPTIME}$-complete by reducing the left-to-right games considered in \cite{Muscholl2005}.

% experiments
\section{Experiments}
\label{sec:experiments}
%To test the perfomrance of the algorithm proposed in this paper, we developed a C++-implementation.
%The code is publically available at \texttt{https://github.com/SebastianMuskalla/RIGG}.
% Todo: put it also on our group page

We have implemented our algorithm in C++ \cite{implementation} and
compared it to an implementation of Cachat's algorithm \cite{Cachat2002} for games on pushdown systems. 
Cachat's input instances consist of a pushdown system $P$ with an ownership partitioning on the control states and an alternating finite automaton over the stack alphabet of $P$ ($P$-AFA). 
%with one initial state per control state of $P$.
%The $P$-AFA accepts a regular set of configurations, and 
The first player wins by enforcing a run into a configuration accepted by the $P$-AFA. 
Cachat's algorithm constructs the winning region of the first player by saturating the automaton. %by adding transitions so that the resulting automaton accepts the winning region of the first player.

To convert instances of our game to that of Cachat, we construct 
%One can convert an instance of our game to an instance of the first game.
a pushdown system $P$ that encodes both the grammar $G$ and the target automaton $A$. 
A sentential form $wX\beta$ (where $X$ is the left-most non-terminal) will be represented in $P$ by a configuration $(Q', X\beta)$, where $Q'$ is the set of states that $A$ can be in after processing $w$. 
To be precise, for each subset $Q' \subseteq Q$ of states of $A$, $P$ has two control states $Q'_\square$ and $Q'_\roundss$, one for each player. 
\begin{enumerate}[(1)]
    \item
        If $P$ is in control state $Q'_\player$, and the topmost stack symbol is a terminal $a \in T$, it can be popped and the control state is changed to $Q''_\player$, where $Q''$ is the set of states with $Q' \tow{a} Q''$.
    \item
        If $P$ is in control state $Q'_\player$, and the topmost stack symbol is a non-terminal $Y$ owned by the other player $\opponent$, $P$ goes to control state $Q'_{\opponent}$ without modifying the stack.
    \item
        If $P$ is in control state $Q'_\player$, and the topmost stack symbol is a non-terminal $X$ owned by $\player$, there is one transition in $P$ for each rule $X \to \alpha$ with $X$ on its left-hand side that pops $X$ and pushes $\alpha$ without changing the control state. 
\end{enumerate}
\raggedbottom
 
The $P$-AFA just checks that the set of states $Q'$ of the current control state $Q'_\player$ does not contain a final state and that the stack is empty.

We have to encode both the automaton and the grammar into the pushdown system. 
If we keep the whole sentential form on the stack, the terminal prefix prevents us from modifying the non-terminals.

The translation thus embeds a determinized version of $A$ in $P$. This may cause an exponential blow-up in the size of the input instance, which reflects the worst case complexity: Our problem is $\mathsf{2EXPTIME}$-complete while Cachat's algorithm is exponential.
%
%One might notice that effectively, we are encoding the determinization of $A$ in $P$.
%One the one hand, this is expected from a complexity-theoretic perspective: Cachat's algorithm works in exponential time, but as we have proven, our problem is $\mathsf{2EXPTIME}$-complete.
%One the other hand, this makes sense from a semantical perspective: In the semantics of our game, the automaton processes the derived word without any interference by the players.
%If we would in each step give the active player the power to resolve the non-determinism in the automaton, this would change the semantics of the game.

%Due to the lack of non-trivial examples, we used random generated automata and grammars.
For the experiments, we generated random automata using the Tabakov-Vardi model~\cite{TabakovV:2005}.
The generator is parameterized in the number of letters and control states, the percentage of final states, and the number of transitions per letter 
(given as a fraction of the number of states).
We adapt the model to generate also grammars with rules of the form $X \to a Y b$, 
with parameters being the number of rules and non-terminals for each player, 
and the chances of $a, Y$, and $b$ to be present.
Since sparse automata and grammars are likely to yield simpler instances, we focus on dense examples.

For the parameters $|Q|, |T|$, $|N_\roundss|$ and $|N_\square|$, we tried out several combinations.
The entry $x/y/z$ in the table below stands for $|Q| = x, |T| = y, |N_\roundss| = |N_\square| = z$. For each combination, we generated $50$ random automata and grammars, applied three algorithms to them, and measured how many instances could be solved within 10 seconds and how much time was consumed for the instances that could be solved on average.\vspace{0.05cm}
We compared:
(1) Our algorithm with a naive Kleene iteration, i.e.\ all components of the current solution are updated in each step.
(2) Our algorithm with chaotic iteration implemented using a worklist, i.e.\ only components whose dependencies have been updated are modified. 
This is the common way of implementing a Kleene iteration. 
(3)~Cachat's algorithm applied to our problem as described above. 
To improve the runtime, the target automaton has been determinized and minimized before creating the pushdown system. We ran our experiments on an Intel i7-6700K, $4$GHz. The durations are milliseconds.
\begin{center}
\newcommand{\showparam}[1]{\texttt{#1}}  
\newcommand{\showavg}[1]{\texttt{#1}}    
\newcommand{\showpct}[1]{\texttt{#1}}  
\scalebox{0.9}{
\begin{tabular}{|c||rr|rr|rr|}
\hline
& \multicolumn{2}{c|}{naive Kleene}
& \multicolumn{2}{c|}{worklist Kleene}
& \multicolumn{2}{c|}{Cachat}\\
\hline
& average time & \% timeout & average time & \% timeout & average time & \% timeout\\
\hline
\hline
\showparam{ 5/ 5/ 5} & \showavg{65.2} & \showpct{2} & \showavg{0.8} & \showpct{0} & \showavg{94.7} & \showpct{0} \\
\hline
\showparam{ 5/ 5/10} & \showavg{5.4} & \showpct{4} & \showavg{7.4} & \showpct{0} & \showavg{701.7} & \showpct{0} \\
\hline
\showparam{ 5/10/ 5} & \showavg{13.9} & \showpct{0} & \showavg{0.3} & \showpct{0} & \showavg{375.7} & \showpct{0} \\
\hline
\showparam{ 5/ 5/15} & \showavg{6.0} & \showpct{0} & \showavg{1.1} & \showpct{0} & \showavg{1618.6} & \showpct{0} \\
\hline
\showparam{ 5/10/10} & \showavg{32.0} & \showpct{2} & \showavg{122.1} & \showpct{0} & \showavg{2214.4} & \showpct{0} \\
\hline
\showparam{ 5/15/ 5} & \showavg{44.5} & \showpct{0} & \showavg{0.2} & \showpct{0} & \showavg{620.7} & \showpct{0} \\
\hline
\showparam{ 5/ 5/20} & \showavg{3.4} & \showpct{0} & \showavg{1.4} & \showpct{0} & \showavg{3434.6} & \showpct{4} \\
\hline
\showparam{ 5/10/15} & \showavg{217.7} & \showpct{0} & \showavg{7.4} & \showpct{0} & \showavg{5263.0} & \showpct{16} \\
\hline
\showparam{10/ 5/ 5} & \showavg{8.8} & \showpct{2} & \showavg{0.6} & \showpct{0} & \showavg{2737.8} & \showpct{2} \\
\hline
\showparam{10/ 5/10} & \showavg{9.0} & \showpct{6} & \showavg{69.8} & \showpct{0} & \showavg{6484.9} & \showpct{66} \\
\hline
\showparam{15/ 5/ 5} & \showavg{30.7} & \showpct{0} & \showavg{0.2} & \showpct{0} & \showavg{5442.4} & \showpct{52} \\
\hline
\showparam{10/10/ 5} & \showavg{9.7} & \showpct{0} & \showavg{0.2} & \showpct{0} & \showavg{7702.1} & \showpct{92} \\
\hline
\showparam{10/15/15} & \showavg{252.3} & \showpct{0} & \showavg{1.9} & \showpct{0} & \showavg{n/a} & \showpct{100} \\
\hline
\showparam{10/15/20} & \showavg{12.9} & \showpct{0} & \showavg{1.8} & \showpct{0} & \showavg{n/a} & \showpct{100} \\
\hline
\end{tabular}
}
\end{center}
Already the naive implementation of Kleene iteration outperforms Cachat's algorithm, which was not able to solve any instance with parameters greater than \texttt{10/15/15}.
The worklist implementation %of Kleene iteration 
is substantially faster, by three orders of magnitude on average.  
This confirms our hypothesis: The stack content is more information than needed for safety verification, and getting rid of it by moving to the summary domain speeds up the analysis.

One can also implement Cachat using a worklist.
Since in every step not only single transitions of the $P$-AFA but whole paths in the $P$-AFA are considered, handling the states of the $P$-AFA with a worklist is not possible.
But it is possible to handle the states of the pushdown system $P$ using a worklist: Whenever a transition labeled by stack symbol $a$ is added to the $P$-AFA, all control states of $P$ that have a transition that pushes $a$ to the stack have to be added to the worklist. 

Unfortunately, this does not help for the instances obtained by  encoding our type of game. 
For every non-terminal $X$ owned by player $\player$
(1) all states $Q'_{\opponent}$ of the of the opponent $\opponent$ have to be added to the worklist, since there is a transition from $Q'_{\opponent}$ to $Q'_\player$ that pops $X$ and pushes $X$, and
(2) if $X$ occurs on the right-hand side of at least one rule of the grammar, say in $Y \to \alpha$, all states $Q'_{\player}$ of the of the player $\player$ have to be added to the worklist, since there is a transition that pops $Y$ and pushes $\alpha$ for each such state.

The terminal symbols are handled in the very first iteration of Cachat.
This means that starting from the second iteration, adding a transition to the $P$-AFA will cause almost all states of $P$ to be added to the worklist.
In our experiments, the worklist variant was slower by at least one order of magnitude, even on small examples.

% antichains
\newcommand{\rejecting}[1]{\mathit{reject}(#1)}
\newcommand{\unfold}{\mathit{unfold}}
\newcommand{\eval}[1]{\mathit{eval}(#1)}

\newcommand{\subsumes}{\sqsubseteq}
\newcommand{\subsequiv}{\equiv}
\newcommand{\compose}{;}
\newcommand{\fwsim}{\preceq_\texttt{fw}}
\newcommand{\bwsim}{\preceq_\texttt{bw}}

\renewcommand{\theequation}{\Roman{equation}}

\section{Algorithmic Considerations}
\label{sec:antichains}
We discuss how to further speed-up the worklist implementation by two heuristics prominent in verification: Lazy evaluation~\cite{fiedor:lazy} and antichains~\cite{Fogarty:Efficient,Abdulla:Simulation,Abdulla:Advanced}. 
The heuristics are not meant to be a contribution of the paper and they are not yet implemented.
The point is to demonstrate that the proposed summary domain combines well with algorithmic techniques. 
For both heuristics, it is not clear to us how to apply them to the domain of alternating automata. 

The idea of \emph{lazy evaluation} is to keep composed formulas $(F\vee F');M$ symbolic, i.e.\ we store them as a term rather than computing the resulting formula.
When having to evaluate the formula represented by the term, we only compute up to the point where the value influences the overall answer. 
Consider the test whether $(F\vee F');M$  is rejecting. 
If already $F;M$ is rejecting, the whole formula represented by the term will be rejecting and the evaluation of $F';M$ can be skipped.
.

The idea of the \emph{antichain optimization} is to identify representative elements in the search space that allow us to draw conclusions about all other elements. 
Here, the search space consists of formulas (representing the intermediate steps of the fixed-point computation).  
By Lemma~\ref{lemma:leq_monotonic}, it is sufficient to reason modulo logical equivalence.  
This allows us to remove redundant disjuncts and conjuncts,
in particular, if $F\lleq G$ we can prune $F$ from $F\lor G$ and $G$ from $F\land G$.  
When reasoning over CNFs,
this removes from a formula all clauses that are subsumed by other clauses. 
It is thus enough to store the CNFs in the form of antichains of $\subseteq$-minimal clauses.
The antichain approach benefits from a weaker notion of implication. 

\subsection{Lazy Evaluation}

Assume we are interested in whether refuter wins from a given sentential form $\alpha$.
For simplicity, introduce a fresh non-terminal $S$ with the single rule $S \to \alpha$.
That refuter wins the game can be concluded as soon as the formula $\sol{S}$ becomes rejecting. 
The fixed-point iteration does not have to continue beyond this point.  
Inspired by \cite{fiedor:lazy}, the idea of early termination upon reaching a target can be elaborated further into a \emph{lazy evaluation} technique: 
Also compositions of formulas have to be evaluated only up to the point where the value influences the overall answer.
Consider the test whether $(F\lor F');M$ is rejecting.
If $F;M$ is already rejecting, then the whole formula is rejecting, 
and the evaluation of $F';M$ can be skipped.

The lazy algorithm evaluates the predicate $\rejecting\bot$, which will yield $\itrue$ iff $\sol S$ is rejecting.
The test is done on the fly, while unfolding the fixed-point computation from $\bot$ ($X\mapsto\ifalse$ for all $X\in N$). 
The algorithm uses a set of rules (discussed below) that reduce the test on their left-hand side into a Boolean combination of tests on the right, and which return a Boolean value on boxes. 

The first rule unfolds the fixed-point computation. 
A reasonable implementation of the algorithm would not use the Kleene iteration on the product domain from Section~\ref{sec:domain}
but a variant of the more efficient chaotic iteration (\eg the worklist algorithm)~\cite{SeidlWilhelmHack2012}.   
The following rule corresponds to one unfolding of chaotic iteration:
\vspace{1mm}
\begin{align}
\label{rule:unfold}
\rejecting{\sigma} \iff& \rejecting{\sigma_S} \lor \unfold\ ,
\\[1mm]
\text{ where } \quad
\unfold =& 
\left\{
\begin{array}{ll}
\rejecting{\sigma[X\mapsto f_X(\sigma)]}
& \text{, for some } X\in N \text{ with } f_X(\sigma) \not\subsumes \sol{X}\ ,
\\[0.1cm]
\ifalse
& \text{, if $f_X(\sigma)\subsumes \sol{X} $ for all $X\in N$}\ .
\end{array}
\right.
\nonumber
\end{align}
\\
Intuitively, either the formula $\sol{S}$ 
is already rejecting in the current assignment, or it may become rejecting in some further unfolding of the fixed-point computation. 
The disjunct $\unfold$
updates one of the variables based on the current assignment $\sigma$, or it terminates the fixed-point computation if no variable can be updated, that is, if the current assignment is already the least fixed point.

To evaluate the first disjunct in Rule~\eqref{rule:unfold}, we have to evaluate $\rejecting{\sol{S}}$ on a term $\sol{S}$ built on top of boxes using Boolean connectives and relational composition. 
To evaluate the second disjunct, we need to decide subsumption. 
On boxes, $\rejecting{\tboxr} = \itrue$ iff $\tboxr$ is rejecting, and $\rho \subsumes \tau$ is evaluated according to the notion of subsumption in use. 
On Boolean connectives, rejection and subsumption queries are evaluated as Boolean combinations of the respective queries over the subformulas.
If the tested formulas are relational compositions, then the composition is first pushed down one level using Definition~\ref{def:cnf_seqcomp}, or, on leaf level, 
it is eliminated by composing the boxes. 
This is expressed by the following rules where
$\rho\in\Boxes_A$, $F,F',M,M',G,G'\in\BF_A$, and $\{\star,\bar\star\} = \{\land,\lor\}$:

\begin{center}
    $
    \begin{aligned}
    \rejecting{F\star G}	&\iff \rejecting{F} \star \rejecting G\ ,
    &\rejecting{F;M} 		&\iff \rejecting{\eval{F;M}}\ ,        \\        			
    F\star F'\sqsubseteq G	&\iff F\sqsubseteq G \mathrel{\bar\star} F'\sqsubseteq G\ ,
    &F;M\mathrel\subsumes G	&\iff \eval{F;M}\mathrel\subsumes G\ ,\\				
    F\sqsubseteq G\star G'	&\iff F\sqsubseteq G  \star F\sqsubseteq G'\ ,
    &F\mathrel\subsumes G;M	&\iff F \mathrel\subsumes \eval{G;M}\ ,
    \end{aligned}
    $  
\end{center}

\begin{center}
    $
    \begin{aligned}
    \text{where}\quad \eval{F\star F'\compose M} =\ & {F\compose M\star F'\compose M}\ ,\\
    \eval{\rho\compose M\star M'} =\ & {\rho\compose M\star \rho\compose M'}\ .
    \end{aligned}
    $  
\end{center}

Using these rules, the algorithm constructs an alternating proof tree rooted at $\rejecting{\bot}$ which branches at Boolean connectives on the right-hand sides of the rules.  
It may conclude before the fixed-point computation terminates, and when it concludes,
the proof tree may still contain not fully evaluated queries and compositions which do not influence the final answer.  
These unfinished computations are efficiency gains over the basic algorithm which always iterates until fixed point and evaluates all compositions.

We note that if the rules were implemented verbatim, the tree would contain repetitive evaluations of the same queries and compositions. 
This could be countered by joining equivalent queries and keeping the proof in the form of a directed acyclic graph. 
Similarly, it is useful to keep occurrences of the same compositions as references to one representative and evaluate all of them simultaneously by evaluating the representative. 

\subsection{Antichains}

The antichain approach benefits from a weaker notion of implication.    
Consider a preorder ${\subsumes}\subseteq\Boxes_A\times\Boxes_A$ that reflects the reject predicate as follows: If $\rho\ \subsumes\ \tau$ and $\rho$ is rejecting, then $\tau$ should be rejecting.   
These implications among boxes justify further implications among formulas.
%% Formally, the \emph{subsumption} of formulas (relative to $\subsumes$) is defined as
%% 
%% \renewcommand{\theequation}{\Roman{equation}}
%% \[
%%     F\subsumes G, \text{\ \ if\ \  }\bigwedge_{\rho,\tau\in\Boxes_A \colon \rho\ \subsumes\ \tau}\rho\lleq \tau\hspace{0.2cm} \models\hspace{0.2cm} F\lleq G\ .
%% \]
Formally, formula $F$ \emph{subsumes} $G$, written $F\subsumes G$, if 
\[
    \Bigg(
        \bigwedge_{\rho,\tau\in\Boxes_A \colon \rho\ \subsumes\ \tau}\rho\lleq \tau
    \Bigg)
    \hspace{0.2cm} \models\hspace{0.2cm}
    F\lleq G\ .
\]
Under the condition that the following analogue of Lemma~\ref{lemma:leq_monotonic} holds, it is possible to reason modulo subsumption instead of implication:
\begin{align}
\label{eq:monotonicity}
\forall F,F',M,M'\in\BF_A \colon
\textit{ If }
F\subsumes F' \textit{ and } M \subsumes M'
\textit{, then }
F;M \subsumes F'; M' \
.
\end{align}
Subsumption was used in \cite{Fogarty:Efficient,Abdulla:Simulation,Abdulla:Advanced},
where a fixed point over sets of boxes is computed to check language inclusion among \Buechi automata. 
For the definition of $\subsumes$, the simplest possibility is to take inclusion  among boxes \cite{Fogarty:Efficient}. 
A more advanced option, following \cite{Abdulla:Advanced}, further weakens the set inclusion using simulation on states of $A$.

In Boolean satisfiability, the antichain optimization corresponds to the subsumption rule, and it is known to have a limited impact on the performance of solvers. 
The setting we consider, however, is different.
Our formulas are enriched during the computation by new clauses (that are not derived from others as in SAT). 
The antichain optimization can therefore be expected to yield better results for inclusion games and, in fact, has been successfully implemented for automata models~\cite{Fogarty:Efficient,Abdulla:Simulation,Abdulla:Advanced}.

% related work
\section{Related Work}
\label{sec:related}
We already discussed the relation with
Cachat's work~\cite{Cachat2002}.
% ----------------------------------
% studies games induced by a pushdown system
%with an ownership partitioning on the control states.
%The goal is to reach a set of configurations defined by an alternating automaton. 
%The contribution is a fixed-point computation of the winning region that saturates the alternating automaton by adding transitions.  
% ----------------------------------
%To compute a strategy, backtracking the accepting runs and selecting a run with minimal cost is needed.
% Reviewer: No, just needed if we want the shortest run
% O. Serre: The part of the paper about the strategy synthesis contains a bug
% ----------------------------------
%Our domain is easier to handle algorithmically. 
%By compositionality, we obtain information for any sentential from which a finite play can be enforced, even for the ones where refuter cannot win.
% ----------------------------------
Walukiewicz \cite{Walukiewicz2001234} studies games given by a pushdown automaton with a parity function on the states. 
Similar to our case, the derived strategies are implementable using a stack. 
The problem \cite{Walukiewicz2001234} is concerned with is different from ours in several respects.  
The game aspect is given by the specification (a $\mu$-calculus formula), not by the system as in our case. 
Moreover, (infinite) parity games are generally harder than safety:
%The different nature of the problem is reflected by the complexity of the decision procedure:
\cite{Walukiewicz2001234} is exponential both in the system and in the specification, while our construction is exponential only in the specification.
Piterman and Vardi \cite{Piterman2004} study a similar variant of the problem and come up with a solution originating in the automata-theoretic approach~\cite{Kupferman2010}.

Walukiewicz reduces solving parity games on the infinite computation tree of a pushdown system to solving parity games on a finite graph.
To do so, instead of the full stack, only the topmost stack symbol is stored.
Whenever a push should be executed, one player guesses the behavior of the game until the corresponding pop, i.e.\ she proposes a set of control states.
The other player can decide to skip the subgame between push and pop by selecting a control state from the set, and the game continues.
Alternatively, she can decide to verify the subgame.
In this case, the new symbol becomes top-of-stack, and the game continues until it is popped.
After the pop, the game ends, and which player wins is dependent on whether the current control state is in the proposed set of states.

This approach can be applied to a context-free game to reduce it to a reachability game on a doubly-exponentially-sized graph.
Before applying a rule to the leftmost non-terminal $X$, we let refuter propose a set of boxes that describes the effect of terminal words derivable from $X$.
Prover can either accept the proposal and select one of the boxes, or she verifies the proposal. In the latter case, the rest of the sentential form can be dropped.
A winning strategy for refuter in the finite game has to guess the effect of each non-terminal, while our method deterministically computes it:
The guessed effects that will not lead to refuter losing the subgame are exactly the sets of boxes occurring as the image of a choice function.

The work \cite{Muscholl2005} considers active context-free games where in each turn, player A picks the position of a non-terminal in the current sentential form and player B picks the rule that is applied to the non-terminal. 
It is shown to be undecidable whether player A can enforce the derivation of a word in a regular language. 
If one limits the moves of player~A to left-to-right strategies (skipped non-terminals cannot be touched again, the regular target language may contain non-terminals), one obtains a game that is closely related to our setting.
In fact, the authors show that allowing player~A to pick the rules for some of the non-terminals does not increase the expressive power.
Therefore, there are polynomial-time reductions of our type of game to their type of game and vice versa.
In~\cite{Muscholl2005}, the focus lies on establishing the lower bounds for the time complexity of various type of active context-free games.
The authors show that deciding the existence of a left-to-right winning strategy is \textsf{2EXPTIME}-complete, like the problem considered in this paper (Section \ref{sec:complexity}).
The upper bound is shown by using an exponential-time reduction to Walukiewicz~\cite{Walukiewicz2001234}, and they also present an optimal algorithm that uses Cachat's algorithm for pushdown systems.
Our algorithm also has optimal time complexity, but contrary to \cite{Muscholl2005}, it is based on procedure summaries rather than on saturation.
The lower bound is shown by encoding an alternating Turing machine with exponential space as a grammar game, and we adapted their proof to show Theorem~\ref{thm:hardness}.
\cite{Muscholl2005} was further elaborated on and generalized in \cite{Bjorklund2013,SchusterS15}.

Methods for solving variants of pushdown games, related mostly to saturation (see \cite{Matthew:saturationalgs} for a survey on saturation-based methods), are implemented in several tools. 
\cite{CSHORE} targets higher-order pushdown systems, 
related to it is \cite{HORSAT}, 
\cite{CERTIFICATES} implements an optimized saturation-based method, 
\cite{PDSOLVER} solves the full case of parity games. 
\cite{PREFACE} implements a type directed algorithm not based on saturation.
None of the tools implements procedure summaries, but some can be used to solve instances of our problem.
%, modulo a reduction of the input instances. 
We plan to carry out a thorough comparison with these implementations in the future.
%
%\rfmcomment{Comment: This paragraph needs context.}
%In \cite{Piterman2004}, a positional strategy is constructed which is represented by a DFA of exponential size in the game reading the current configuration. 
%A similar result is obtained in \cite{Matthew:regular} using the saturation method. 

%In \cite{Schuster2015}, the games from \cite{Muscholl2005} are extended to nested words.
%We hope that this can also be done in our setting, e.g. by considering the boxes that are used for testing inclusion among visibly-pushdown languages in \cite{Friedmann2013}.

%The computation of a least fixed point for a system of equations interpreted over logical formulas is also studied in \cite{Mader1995}. 
%The work factorizes along logical equivalence, like we do. Its different domain of Boolean variables, however, makes the main idea, Gaussian elimination, not applicable in our setting. 

Antichain heuristics, discussed in Section~\ref{sec:antichains}, 
were developed in the context of finite automata and games \cite{Wulf2006,WulfDR06},
and generalized to B\"uchi automata \cite{Fogarty:Efficient,Abdulla:Simulation,Abdulla:Advanced} with a fixed point over sets of boxes. 
Our lazy evaluation is inspired by \cite{fiedor:lazy}. 
Our framework is compatible with techniques for reachability in well-structured transition systems (WSTS) that proceed backwards~\cite{AbdullaCJT96}. 
We believe that techniques like \cite{Kloos2013,Ganjei2016,Kaiser2012,Abdulla2013,Ganty2006} can be adapted to our setting. 
To instantiate general WSTS reachability algorithms, 
the ordering of configurations would be based on implication among formulas, the target set would be the upward closure of the assignment $\sigma$ where $\sigma_S$ is the conjunction of all rejecting boxes and $\sigma_X = \ifalse$ for every other $X\in N$, and the initial state would be the assignment $\bot$.
Another interesting possibility would be to adapt Newton iteration \cite{EsparzaKieferLuttenberger2010}.

The transition monoid can be traced back at least to B\"uchi~\cite{Buchi1990}, and was prominently used \eg in \cite{Vardi85}.

%% Besides the mentioned applications in inclusion checking and numerous theoretical studies, 
%% it has since appeared in a variety of practically-oriented works where efficiency is a primary concern.
%% Examples are size-change termination~\cite{Lee2001} and procedure summaries~\cite{SharirPnueli1978,RepsHorwitzSagiv1995}.
%% The idea of precomputing effects of some program/automaton steps as a relation over its states, %corresponding to boxes, 
%% and composing this information, is also the essence of the works \cite{Mytkowicz:2014,Raychev:2015} on (automatic) parallelization.

\raggedbottom

% bibliography
\newpage
\bibliography{cited}

\end{document}